\documentclass[11pt,letterpaper]{article}
\usepackage[dvipsnames]{xcolor}
\usepackage{amsfonts}
\usepackage{amsmath}
\usepackage{amssymb}
\usepackage{amsthm}
\usepackage{mathtools}
\usepackage{tikz}
\usepackage{pgfplots}
\pgfplotsset{compat=1.18}
\usetikzlibrary{decorations.pathreplacing}
\usetikzlibrary{patterns}
\usepackage{enumitem}
\usepackage{graphicx}
\usepackage[ruled,vlined,linesnumbered]{algorithm2e}
\providecommand{\DontPrintSemicolon}{\dontprintsemicolon}
\usepackage{xspace}
\usepackage{booktabs}
\usepackage{authblk}
\usepackage{nicefrac}
\usepackage{parskip}

\usepackage[square, semicolon]{natbib}
\usepackage{geometry}
\usepackage[colorlinks = true, linkcolor=NavyBlue,citecolor=ForestGreen]{hyperref}
\usepackage[nameinlink]{cleveref}

\usepackage[tt=false]{libertine}

\newtheorem{theorem}{Theorem}
\newtheorem{lemma}{Lemma}
\newtheorem{definition}{Definition}

\newtheorem{corollary}{Corollary}
\newtheorem{claim}{Claim}
\newtheorem{remark}{Remark}
\newtheorem{proposition}{Proposition}

\newcommand{\bmu}{\ensuremath{\boldsymbol{\mu}}\xspace}

\newcommand{\mmsT}{\ensuremath{\textsc{mms}}\xspace}

\newcommand{\gmms}{\ensuremath{\textsc{gmms}}\xspace}

\newcommand{\mms}{\ensuremath{\boldsymbol{\mu}}\xspace}

\newcommand{\efo}{\ensuremath{\textsc{ef}\oldstylenums{1}}\xspace}
\newcommand{\efx}{\ensuremath{\textsc{efx}}\xspace}

\newcommand{\nsw}{\texttt{NSW}}
\newcommand{\pool}{\ensuremath{\mathcal{P}}\xspace}
\newcommand{\prob}{problematic\xspace}

\DeclareMathOperator*{\argmin}{arg\,min}
\DeclareMathOperator*{\argmax}{arg\,max}

\title{Improved Approximation Guarantees for\\ Groupwise Maximin Share Fairness\thanks{An extended abstract version of this work was accepted to the 19th International Symposium on Algorithmic Game Theory (SAGT 2026).}}

\author[1,2]{Georgios Amanatidis}
\author[1]{Anna Korfiati}
\author[1,2,3]{Evangelos Markakis}
\author[1,2]{Christodoulos Santorinaios}

\affil[1]{Athens University of Economics and Business, Greece}
\affil[2]{Archimedes/Athena RC, Greece}
\affil[3]{Input Output Group (IOG), Greece}

\hypersetup{
pdftitle={Improved Approximation Guarantees for Groupwise Maximin Share Fairness},
pdfsubject={},
pdfauthor={Georgios Amanatidis, Anna Korfiati, Evangelos Markakis, Christodoulos Santorinaios},
pdfkeywords={Groupwise Maximin Share, Discrete Fair Division, Approximation Algorithms.},
}

\begin{document}
\maketitle
\begin{abstract}
    We study the problem of fairly allocating a set of indivisible goods to a set of $n$ agents with additive valuation functions. We focus on the very demanding notion of \emph{groupwise maximin share fairness} (\gmms), which requires that each agent $i$ receives value comparable to her maximin share, where the latter is computed \emph{with respect to any subset of agents that contains $i$}. We show that it is possible to compute $(\phi-1)$-approximate \gmms allocations in polynomial time, where $\phi \approx 1.618$ is the golden ratio. This improves on the previously known guarantee of $4/7$ of \cite{chaudhury2021little} and \cite{amanatidis2020multiple}. We propose a simple algorithm that maintains the same main properties as the Draft-and-Eliminate algorithm of \cite{amanatidis2020multiple} and we improve on the approximation guarantee analysis by carefully bounding the maximin share value within any subinstance induced by the restriction of our allocation to a subset of agents. Our analysis is asymptotically tight for algorithms that share these properties and has the additional benefit of giving improved guarantees for restricted settings; in particular, when the agents agree on the top-$n$ goods or when the number of agents is small. To illustrate the challenges of going beyond the guarantees of our algorithm, we also present a variant with an improved approximation of $(\sqrt{10}-1)/3 \approx 0.72$ for the case of three agents. To achieve this improvement we partially characterize the maximin share guarantees of short picking sequences for a small number of goods.
\end{abstract}

\section{Introduction}

 Fair division is concerned with the problem of partitioning a set of resources among agents with diverse preferences in a manner that is considered \emph{fair} by the agents. The study of this fundamental problem has a rich history, originating in the work of \cite{Steinhaus49} (with the contribution of Banach and Knaster)  and has since attracted sustained attention across economics, mathematics, and computer science. While the early literature focused primarily on the allocation of \emph{divisible} resources---where proportionality \citep{Steinhaus49} and envy-freeness \citep{GS58,Varian74} serve as the central fairness benchmarks---the fair allocation of \emph{indivisible} goods has emerged as a highly active research frontier in recent years (see, e.g., \cite{AmanatidisABFLMVW23}).

In the indivisible setting, strong fairness guarantees are in general unattainable: when a single valuable item must be assigned to one of several agents, no allocation can be fair to all. This has motivated the study of appropriate \emph{relaxations}, most notably \emph{envy-freeness up to one good} (\efo) \citep{LMMS04,Budish11}, \emph{envy-freeness up to any good} (\efx) \citep{CaragiannisKMPS19}, and \emph{maximin share fairness} (\mmsT) \citep{Budish11}. In particular, the maximin share of an agent is the best value she can guarantee to herself by partitioning the goods into $n$ bundles (where $n$ is the total number of agents) and receiving her least preferred one; an \mmsT allocation assigns every agent at least her maximin share. While such allocations need not always exist \citep{kurokawa2018fair,feige2021tight}, a flourishing line of work has established that constant-factor approximations are always achievable \citep{kurokawa2018fair,HeidariKSS26,HuangZ25}. A natural strengthening of \mmsT is \emph{groupwise maximin share fairness} (\gmms), introduced by \cite{BBMN18}, which demands that the \mmsT guarantee holds not only with respect to the whole set but within every subgroup of agents. This stronger requirement captures collective fairness on a richer level, yet its approximability is significantly less understood. In this paper, we study the existence and efficient computation of approximate \gmms allocations for agents with additive valuation functions.

Note that even checking whether an allocation is approximately \gmms seems nontrivial, and this is not just because computing the maximin shares is hard; even if we had an oracle for the latter, it is still unclear how to efficiently verify that \mmsT guarantees hold for every single subset of agents. Given this, it is surprising that there are any positive results on the existence and computation of allocations that approximately satisfy \gmms fairness. In particular, \cite{BBMN18} showed how to efficiently compute $1/2$-\gmms allocations, and this approximation factor was later improved to $4 / 7$ \citep{chaudhury2021little,amanatidis2020multiple}. Since then, it has remained open whether better approximations are possible.
\smallskip

\noindent\textbf{Our Contributions.}
We show that, for agents with additive valuation functions, $(\phi - 1)$-\gmms allocations always exist and can be efficiently computed. For this, we revisit the algorithm of \cite{amanatidis2020multiple}, along with a closely related algorithm for computing approximately \efx allocations by \cite{Farhadietal21}, in order to present a simpler version with the desirable  approximation guarantee with respect to \gmms.
For restrictions of the setting, we obtain significant improvements. In particular, 
\begin{itemize}[leftmargin=20pt,itemsep=4pt]   
\item 
We present an algorithm that computes $\min\left\{\frac{2}{3}, \frac{\sqrt{5n^2 - 4n} - n}{2n-2}\right\}$-\gmms
allocations (\Cref{th:main}). For $n\in\{3, 4\}$ this factor is $2/3$ and as $n$ grows it gracefully degrades to $\phi-1\approx 0.618$.
The allocations computed by our algorithm are also \efo and---with a small parameter adjustment---$(\phi-1)$-\efx (\Cref{thm:main}). 
\item 
When all $n$ agents agree on the set of the best $n$ items, we improve the
ratio to $2/3$ (\Cref{thm:topn}), matching the known ratio for \efx for this case. This 
setting  generalizes two particularly well-studied special cases:  i) identical agents, and ii) agents with a common ranking over all items. 
\item 
For the fundamental case of three agents, we show that it is possible to go beyond $2/3$ and efficiently compute $0.72$-\gmms allocations (\Cref{thm:three_agents}). 
\end{itemize}

\smallskip 

\noindent\textbf{Technical Considerations.} Our main algorithm itself (\Cref{alg:main}) is relatively simple, following the general framework
of the state-of-the-art algorithms for computing approximately \efx allocations. 
Our main contribution here is our novel analysis that tightly bounds the proportional share after performing a series of reductions on an extended set of ``problematic'' items.
Interestingly, our analysis can be transferred with small adaptations to either one 
of the two known efficient algorithms for computing $(\phi-1)$-\efx allocations.

In our proofs we often use the \textit{proportional share} (i.e., the average value) of an agent, restricted to any given subset of agents and the union of their bundles, as a \textit{proxy} quantity for her corresponding maximin share value. This alone, however, is too weak to properly bound the performance of the algorithm. As an additional step, we use known monotonicity properties of the maximin share to reduce our instances to ``nicer'' ones. The most standard such reduction is to ignore all agents who have a bundle consisting of a single good \citep{BL16}, but even this cannot show a factor strictly greater than $1/2$. What allowed \cite{chaudhury2021little} and \cite{amanatidis2020multiple} to obtain a guarantee of $4/7$ is a technique of \cite{AmanatidisBM18} that removes one agent at a time along with \emph{two} ``bad'' items as long as there are sufficiently many such items. Crucially, we broaden the definition of such items (that we call \emph{problematic} in the proof of \Cref{lem:other-apx}). As a result, we obtain reduced instances with better structure. Further, we use a different, much more refined, approach with respect to bounding the total value of a reduced instance, leading to a tight analysis of \Cref{alg:main} (\Cref{th:upper}).

To improve our results for $n=3$, we introduce a very particular modification of the envy graph that depends on the 
initial allocation of our algorithm and we perform a thorough case analysis on the possible 
configurations with two or three agents. As an intermediate  result, we characterize the \mmsT guarantees of 
allocating a small number of items via specific picking sequences.

\medskip

\noindent\textbf{Related Work.} \cite{GS58} were the first to introduce the problem of envy-freeness, which was more formally defined later on by \cite{Foley67} and \cite{Varian74}. The first relaxations, \mmsT and \efo, were both defined by \cite{Budish11} while the stronger relaxations, \gmms and \efx, were defined by \cite{BBMN18} and \cite{CaragiannisKMPS19} respectively. 
\cite{LMMS04} implicitly defined \efo and designed the \textit{envy cycle elimination} algorithm to compute it in polynomial time. On the contrary, \mmsT allocations were shown not to exist in \cite{kurokawa2018fair}, whereas the current upper bounds are $\frac{39}{40}$ for 3 agents and $1 - \frac{1}{n^4}$ for a general number of agents, both due to \cite{feige2021tight}. Both bounds apply to \gmms allocations, and are currently the best known. Regarding positive results, \cite{BBMN18} provided a $1/2$ approximation factor which was improved to the current state of the art of $4 / 7$ independently by \cite{amanatidis2020multiple} and \cite{chaudhury2021little}. \smallskip

\noindent\textit{Ordered and top-$n$ instances.} Ordered instances, where the agents agree on a universal ranking of the items, induce a well-motivated setting that often offers an avenue for improved positive results. In the context of \efx allocations, \cite{PR18} showed existence for additive valuations while for \mmsT the existence in general settings reduces to the existence in ordered ones \citep{BL16}. However, the former does not imply a factor better than  $4/7$ for \gmms while the analogous reduction for \gmms is not known to hold. A top-$n$ instance is a mild relaxation of an ordered one where the $n$ agents agree on the set of the top-$n$ items, but not necessarily their order. It was introduced by \cite{MS23} and has led to improved positive results in a variety of fair division problems \citep{ChristoforidisS24,MMP25,AR26}.

\section{Preliminaries}\label{sec:prelims}

A discrete fair division instance can be described by the triplet $(N, M, V)$ where $N =[n]= \{1, \dots, n\}$ is a set of $n$ agents and $M$ a set of $m$ indivisible items. The set $V = \{v_1, \dots, v_n\}$ corresponds to the \emph{additive} valuation functions of the agents, i.e., $v_i(S) = \sum_{g \in S} v_i(g)$. We assume that all the items are \emph{goods}, i.e., $v_i(g) \ge 0$ for every agent $i \in N$ and item $g \in M$.  
Hence, the valuation functions are also \emph{monotone}. For the sake of readability, in what follows, we write $v_{i}(g)$ instead of $v_i(\{g \})$, for $g\in M$.

A complete \emph{allocation} of $M$ to the $n$ agents is a partition, $\mathcal{A} = (A_1,\ldots,A_n)$, where $A_i\cap A_j = \emptyset$ and $\cup_i A_i = M$.
By $\Pi_n(M)$ we denote  the set of all partitions of a set $M$ into $n$ bundles. If $\mathcal{A}$ is a partial allocation we use $\pool(\mathcal{A}) := M \setminus \bigcup_{i \in N} A_i$ to denote the \emph{pool} of unallocated items.
As a first step in our main algorithm, we  use a partial allocation $\mathcal{F}$, such that $|F_i| \le 1$, for all $i\in N$, and the product 
\[\prod_{i\in N :\, |F_i| = 1} \!\!\!\! v_i(F_i)\] 
is maximized.
We refer to $\mathcal{F}$ as a \textit{NSW-maximizing matching allocation}, as it corresponds to a matching between agents and items that maximizes the Nash social welfare. Such allocations can be computed efficiently \citep{Farhadietal21}. 

In fair division, most fairness notions can be categorized into two large families: \emph{share-based} and \emph{envy-based} notions. In the former category, agents perceive fairness by comparing their bundles to the complete set of items while in the latter they compare bundles with each other. As our work focuses on a notion that combines both aspects in a very strong sense, we provide below the necessary definitions from both areas, as well as some tools originally developed for envy-based notions.

\subsection{Share-based fairness notions}

\begin{definition}
	\label{def:mmshare}
	Given $n$ agents, and a subset  $S\subseteq M$ of goods, the $n$-maximin share of agent $i$ with respect to $S$ is $\bmu_i(n, S) = \displaystyle\max_{\mathcal{A}\in\Pi_n(S)} \min_{A_j\in \mathcal{A}} v_i(A_j)$.
\end{definition}
From the definition, it  follows that $n\cdot \bmu_i(n, S)\le v_i(S)$.
When $S=M$, this quantity is just called the \emph{maximin share} of agent $i$. 
An allocation $\mathcal{A}$ is an \emph{$n$-maximin share defining partition} for agent $i$, if $\min_{A_j\in \mathcal{A}} v_i(A_j) = \bmu_i(n, M)$.
When it is clear from context what $n$ and $M$ are, we simply write $\bmu_i$ instead of $\bmu_i(n, M)$. 
\emph{Maximin share fairness}
asks for a partition that gives each agent her (approximate) maximin share.
\begin{definition}
	\label{def:MMS}
	An allocation 
	$\mathcal{A} = (A_1,\ldots,A_n) $ is called an $\alpha$-\mmsT ($\alpha$-maximin share) allocation if $v_i(A_i)\geq \alpha\cdot \bmu_i\,$, for every $i\in N$.
\end{definition} 

Variations of maximin share fairness have also been proposed. Here, we focus on a strengthening of the notion, by demanding an allocation to have an \mmsT-type guarantee for \emph{any subset} of agents. That is, for an agent $i$ and any $S$ with $i\in S$, we can think of $i$ as considering the merged bundle of all items received by members of $S$ and requesting to receive at least her $|S|$-maximin share of this bundle as if it were reallocated among $S$. This is referred to as \emph{groupwise maximin share fairness}, introduced by \cite{BBMN18}.

\begin{definition}
	\label{def:GMMS}
	An allocation $\mathcal{A} = (A_1,\ldots,A_n)$ is called an $\alpha$-\gmms ($\alpha$-groupwise maximin share) allocation if for every subset  of agents $N'\subseteq N$ and any agent $i\in N'$, 
	$v_i(A_i)\geq \alpha\cdot\bmu_i(|N'|, \cup_{j\in N'} A_j )$.
\end{definition} 

In \Cref{def:MMS,def:GMMS}, when $\alpha=1$, we refer to the corresponding allocations as \mmsT and \gmms allocations respectively.
It is clear that the notion of \gmms is significantly stronger than \mmsT.

Before proceeding to the envy-based notions we present some useful lemmas we will repeatedly invoke in the sequel. The first is a monotonicity lemma that allows removing agents and items from an instance, without reducing the maximin share of an agent. Similar statements are standard in the literature of \mmsT allocations and are used to reduce general instances to more structured ones.  

\begin{lemma}[\cite{AmanatidisBM18}]\label{lem:monotonicity}
	Suppose $\mathcal{A} \in \Pi_n(M)$ is an $n$-maximin share defining partition for agent $i$.
	Then, for any set of goods $S$, such that there exists some $j$ with $S \subseteq A_j$, it holds that $\mms_i(n-1, M \setminus S) \geq \mms_i(n, M)$.
\end{lemma} 

An immediate corollary of this reduction lemma is that singleton bundles do not pose an obstacle to obtaining \gmms approximations. The latter also follows from a monotonicity lemma of \cite{BL16}, which is a special case of \Cref{lem:monotonicity}. 

We conclude this section with another simple monotonicity lemma. 
Intuitively, if we start with an $\alpha$-\mmsT allocation and only agent $i$ receives new items, the allocation remains $\alpha$-\mmsT from agent $i$'s perspective.

\begin{lemma}\label{lem:aMMSplusg}
Let $\mathcal{A} = (A_1, \dots, A_n)$ be an allocation of $M' = \cup_{j\in N} A_j \subsetneq M$ such that $v_i(A_i) \ge \alpha \cdot\mms_i(n, M')$ for some agent $i\in N$ and some $\alpha\in (0,1]$.
Then, for any $g \in \pool(\mathcal{A})$, we have $v_i(A_i \cup \{g\}) \ge \alpha \cdot \mms_i(n, M' \cup \{g\})$.
\end{lemma}

\begin{proof}
Suppose, towards a contradiction, that this is not the case, i.e., 
$v_i(A_i \cup \{g\}) < \alpha \cdot \mms_i(n, M' \cup \{g\})$.

By the given lower bound on $v_i(A_i)$ and additivity, this implies:
\begin{equation}
\alpha \cdot \mms_i(n, M') + v_i(g) \le v_i(A_i) + v_i(g)  = v_i(A_i \cup \{g\}) < \alpha \cdot \mms_i(n, M' \cup \{g\}). \label{eq:hyp_1}
\end{equation}
Now, consider an $n$-maximin share defining partition $\mathcal{\bar{A}}$ of $M'\cup \{g\}$ for agent $i$.
Assume, without loss of generality, that $g \in \bar{A}_1$.
We construct a partition $\mathcal{B} = (B_1, \dots, B_n)$ of $M'$ by setting $B_1 = \bar{A}_1 \setminus \{g\}$ and $B_j = \bar{A}_j$ for all $j \neq 1$. For $\mathcal{B}$ we have:
\begin{itemize}[itemsep=5pt]
    \item $v_i(B_1) = v_i(\bar{A}_1) - v_i(g) \ge \mms_i(n, M' \cup \{g\}) - v_i(g)$
    \item $v_i(B_j) = v_i(\bar{A}_j) \ge \mms_i(n, M' \cup \{g\}), \quad \forall j \in N\setminus\{1\}$
\end{itemize}
Thus, $\min_{j\in N}v_i(B_j) \ge \mms_i(n, M' \cup \{g\}) - v_i(g)$. The existence of $\mathcal{B}$ implies $\mms_i(n, M') \ge \mms_i(n, M' \cup \{g\}) - v_i(g)$ and, therefore,
\[\alpha \cdot \mms_i(n, M' \cup \{g\}) \le \alpha \cdot \mms_i(n, M') + \alpha \, v_i(g) \le \alpha \cdot \mms_i(n, M') + v_i(g), \]
which contradicts \eqref{eq:hyp_1}.
\end{proof}

\subsection{Envy-based fairness notions}

We begin our exposition with the notion of envy.
\begin{definition}[$\alpha$-envy]
    For any $\alpha \ge 1$, we say that agent $i$ $\alpha$-envies agent $j$ if $v_i(A_j) > \alpha \cdot v_i(A_i)$.
\end{definition}
When $\alpha = 1$, we simply say that agent $i$ envies agent $j$. To address the issue of allocating an indivisible item to two (or more) agents, envy-based notions involve the hypothetical removal of item(s) from agents' bundles. This is reflected in the following definition.

\begin{definition}\label{def:EF1-EFX}
	An allocation $\mathcal{A} = (A_1,\ldots,A_n)$ is an
	\begin{enumerate}[leftmargin=*,itemsep=3pt,topsep=2pt,parsep=0pt,partopsep=0pt,label=\rm{\alph*})]
		\item $\alpha$-\efo allocation ($\alpha$-envy-free up to one good), if for every pair $i, j\in N$, with $A_j\neq\emptyset$, there exists a good $g\in A_j$, such that
		$v_i(A_i) \geq \alpha\cdot v_i(A_j\setminus \{g\})$.
		\item  $\alpha$-\efx allocation ($\alpha$-envy-free up to any good), if for every pair $i, j\in N$, with $A_j\neq\emptyset$ and every good $g\in A_j$, it holds that $v_i(A_i) \geq \alpha\cdot v_i(A_j\setminus \{g\})$.\label{def:EFX} 
	\end{enumerate}
\end{definition}
Of course, for $\alpha=1$ we obtain precisely the notions of envy-freeness up to one good (\efo) \citep{Budish11,LMMS04} and envy-freeness up to any good (\efx) \citep{CaragiannisKMPS19}. In words, $\alpha$-\efx dictates that agent $i$ does not $\alpha$-envy any proper subset of agent's $j$ bundle while $\alpha$-\efo is satisfied if the envy vanishes for some maximal proper  subset.

It is well known that \efo allocations can be computed in polynomial time, by \cite{LMMS04}, with the celebrated Envy-Cycle-Elimination algorithm (\Cref{alg:ece}), which we also use as a crucial algorithmic component. This is by now a standard algorithm in the literature of fair division but, for the sake of completeness, we provide here a brief presentation and refer to \Cref{app:ece} for the full details. A central concept of this algorithm is the \emph{envy graph}. Given a partial allocation $\mathcal{A} = (A_1,\ldots,A_n)$, the envy graph is a directed graph $G_{\mathcal{A}} = (N, E_{\mathcal{A}})$, where $(i,j) \in E_{\mathcal{A}}$ if and only if agent $i$ currently envies agent $j$, i.e., $v_i(A_i)<v_i(A_j)$.  The key observation of the algorithm is that a source node in $G_{\mathcal{A}}$ corresponds to a non-envied agent in the partial allocation $\mathcal{A}$, and this agent can receive one item without disrupting the \efo property. If no such agent exists, then there is an envy cycle. By repeatedly reallocating the bundles to the agents backwards along envy cycles, the valuation of each agent (weakly) improves, until a source node eventually arises. The fact that the agents' valuations do not decrease is crystallized in the following statement.

\begin{theorem}[Follows by \cite{LMMS04}]\label{thm:ece}
	Let $A_i$ be the bundle assigned to an agent $i\in N$ at the beginning of some iteration of the Envy-Cycle-Elimination algorithm. If $A_i'$ is  assigned to $i$ at the end of any future iteration, then $v_i(A_i') \geq v_i(A_i)$.
\end{theorem}

We will frequently need to reason about an agent's value compared to some good in the pool of unallocated items. For this matter, the next definition will be helpful.

\begin{definition}[$\alpha$-content]
\label{def:content}
    For a partial allocation $\mathcal{A}$ and any $\alpha \ge 1$, we say that agent $i$ is $\alpha$-content with respect to $\mathcal{A}$ if $v_i(A_i) \ge \alpha \cdot v_i(g)$ for any $g \in \pool(\mathcal{A})$. 
\end{definition} 

Throughout the execution of any of our algorithms, we assume that whenever the allocation $\mathcal{A}$ changes, $\pool(\mathcal{A})$ and $G_{\mathcal{A}}$ are silently updated.

\section{A Simple Algorithm and a $(\phi-1)$-\gmms Guarantee}
\label{sec:main}

In order to obtain an improved approximation guarantee for \gmms, we revisit the two known algorithms with the currently best guarantee for \efx, namely Algorithm 3 in \cite{amanatidis2020multiple} and Algorithm 2 in \cite{Farhadietal21}\footnote{We refer to Algorithm 2 in the arxiv version of their work.}. We present Match-Draft-and-Eliminate (\Cref{alg:main}), a simplified hybrid version of these two algorithms.

We briefly explain the formal description of the algorithm. Within the pseudocode, we refer to three (partial) allocations: $\mathcal{F}, \mathcal{S}$ and $\mathcal{A}$. The first one, $\mathcal{F}$, is a NSW-maximizing matching allocation produced in \hyperref[line:1]{line \ref{line:1}}. Recall that in such an $\mathcal{F}$ each agent receives at most one item and, because the product of values is maximized, $G_{\mathcal{F}}$ must be acyclic. For the sake of presentation, we assume that $\mathcal{F}$ assigns to each agent exactly one item. To see that this is without loss of generality, note that any agent $i$ with $F_i = \emptyset$ values positively less than $n$ goods and each of those goods is allocated to a different agent via $\mathcal{F}$; so, any subsequent allocation would be \gmms from the perspective of agent $i$, even if she receives no items. Thus, we may safely ignore such agents; if this includes all agents, then no good has a positive value for any agent and any allocation is \gmms.

Then, our algorithm identifies the set  $R$ of agents who are not reachable in $G_{\mathcal{F}}$ from any  $\theta$-envied agent, and allocates to them one more item, according to their preferences. It is crucial that agents who are envious themselves get their second item early on, hence we follow a topological ordering of a restriction of the envy graph. This step is completed after \hyperref[line:8]{line \ref{line:8}} and forms the second allocation, $\mathcal{S}=(S_1,\ldots,S_n)$. Thus, so far each agent has received at most two items. Then, the final allocation, $\mathcal{A}$, is computed via the Envy-Cycle-Elimination algorithm of \cite{LMMS04} run on $\mathcal{S}$ and $\mathcal{P}(\mathcal{S})$.
Any ties during the execution of \Cref{alg:main} are resolved lexicographically.

\noindent {\em Similarities and differences of \Cref{alg:main} with \cite{amanatidis2020multiple} and \cite{Farhadietal21}.} All algorithms share the same structure, in the sense that they are divided into three phases. First, they allocate one item to each agent. In the second phase they identify a set of agents $R$ that receives a second item. The final phase is to proceed with Envy-Cycle-Elimination. Conceptually, the algorithm we analyze is closer to Algorithm 2 of \cite{Farhadietal21}, albeit simpler in its description, as it avoids the nontrivial notion of envy-rank. Moreover, the second phase is simpler and more similar to the algorithm of \cite{amanatidis2020multiple}. Namely, for a suitable renaming of the agents, our allocation $\mathcal{S}$ can be constructed using two (partial) passes of the simple Round Robin algorithm.\medskip

\begin{algorithm}[H]
	\DontPrintSemicolon 
    {\normalsize
		Let $\mathcal{F} = (\{f_1\}, \dots, \{f_n\})$ be a NSW-maximizing matching allocation\label{line:1}\;
        Set $R = \{i\in N:\text{agent } i \text{ is not reachable in $G_{\mathcal{F}}$ from a } \theta\text{-envied agent} \}$ \;
        Let $\sigma$ be a topological ordering of the envy graph induced by agents in $R$\label{line:order}\;
        $\mathcal{S} = \mathcal{F}$\;
        \For{each agent $i\in R$ with respect to $\sigma$\label{line:4}}{
            $s_i = \argmax_{g \in \pool(\mathcal{S})} v_i(g)$\;
            $S_i = \{f_i\} \cup \{s_i\}$\label{line:6}
        }
        \For{each agent $i \not\in R$}
        {
            $S_i = \{f_i\}$\label{line:8}
        }
        $\mathcal{A} = \text{Envy-Cycle-Elimination}(N, \mathcal{S}, \pool(\mathcal{S}))$\label{line:ece}\;
		\Return $\mathcal{A}$ \;
        }
	\caption{Match-Draft-and-Eliminate$(N, M, \theta)$} \label{alg:main}
\end{algorithm}\medskip

Our main result is the following theorem.

\begin{theorem}\label{thm:main}
Match-Draft-and-Eliminate$(N, M, \phi)$ returns in polynomial time an allocation $\mathcal{A}$ which is $(\phi-1)$-\gmms, \efo, and $(\phi-1)$-\efx.
\end{theorem}

We break down the proof of the \gmms guarantee into two lemmata. In \Cref{lem:ece-apx} we focus on the easier case of agents in $N\setminus R$, who start as $\theta$-content. The trickier case of agents in $R$ is treated in \Cref{lem:other-apx} and uses all the machinery about monotonicity of the maximin share and reductions. For both cases, the analysis is carried out parametrically, based on $\theta$. Then, in the proof of \Cref{thm:main} (as well as in the proof of \Cref{th:main}) we combine the two cases and  balance the resulting factors to obtain our guarantees.

\begin{remark}
    Our analysis on the approximation guarantees of \Cref{thm:main} is also valid for the two aforementioned algorithms in \cite{amanatidis2020multiple} and \cite{Farhadietal21}. We have chosen to present it for our version, as the presentation is simpler.
\end{remark}

We begin our analysis with a simple lemma regarding the properties of a NSW-maximizing matching allocation. 

\begin{lemma}
\label{lem:NSW}
Let $\mathcal{F}$ be a NSW-maximizing matching allocation. For $\theta \ge 1$ and $i,j\in N$ such that agent $i$ $\theta$-envies agent $j$, we have that any agent $k$ reachable from agent $j$ in the envy graph $G_{\mathcal{F}}$ is $\theta$-content. 
\end{lemma}

\begin{proof}
    Let $\mathcal{F} = (\{f_1\}, \dots, \{f_n\})$ be a maximal NSW matching allocation, and suppose agent $i$ $\theta$-envies agent $j$ and agent $k$ is reachable from agent $j$ via the envy path $P:j = a_0 \to a_1 \to \dots \to a_L = k$ (where, possibly, $L = 0$ and $j = k$). Assume for the sake of contradiction that agent $k$ is not $\theta$-content. That is, there exists some item $s_k \in \pool(\mathcal{F})$ such that $v_k(f_k) < \theta \cdot v_k(s_k)$.
    Consider the allocation $\mathcal{F}^\prime  = (\{f'_1\}, \dots, \{f'_n\})$ where $f^\prime_i = f_j, f^\prime_{a_\ell} = f_{a_{\ell+1}} \text{ for } \ell\in\{0,\dots, L-1\}$ (if $L>0$), $f^\prime_k = s_k$ and $f^\prime_\ell = f_\ell$ for $\ell\not\in \{i\} \cup P$. By the assumption that $i$ $\theta$-envies $j$, we get that 
    \begin{align*}
    \nsw(\mathcal{F}^\prime) &= \prod_\ell v_\ell(f^\prime_\ell) =   v_i(f^\prime_i) \cdot v_k(f^\prime_k)  \cdot \prod_{\ell \in P \setminus \{k\}}\!\!\! v_\ell(f^\prime_\ell) \cdot \prod_{\ell \not\in \{i\} \cup P}\!\!\! v_\ell(f^\prime_\ell)\\ 
    &\ge v_i(f_j) \cdot v_k(s_k)   \cdot  \prod_{\ell \ne i,k} v_\ell(f_\ell) >   [\theta \cdot v_i(f_i)] \cdot \Big[\frac{1}{\theta} \cdot v_k(f_k)\Big] \cdot \prod_{\ell \ne i,k} v_\ell(f_\ell)\\
    &= \prod_\ell v_\ell(f_\ell)  =  \nsw (\mathcal{F}) ,
    \end{align*}
    where the first inequality is due to the Pareto improvement along the envy path and the second follows directly from the definitions of $\theta$-envious and $\theta$-content. Hence, $\mathcal{F}^\prime$ contradicts the optimality of the matching that defines $\mathcal{F}$.
\end{proof}

Next, we state a corollary of the lemma that translates the properties of a NSW-maximizing matching in the context of our algorithm. The first and third part of this lemma form the analog of Lemma 4.1 of \cite{amanatidis2020multiple} (up to the use of strict inequalities).

\begin{corollary}\label{cor:sanpoulia}
If $\mathcal{S}$ is the allocation completed in \hyperref[line:8]{line \ref{line:8}} of \Cref{alg:main}, we have  
    \begin{enumerate}[label=\alph*),topsep=5pt,itemsep=4pt]
        \item $v_i(f_i) \ge \theta \cdot v_i(g)$, for any $i \not \in R$ and $g \in \pool(\mathcal{F})$,
        \item $v_i(f_i) \ge v_i(f_j)$, for any $i \not \in R$ and $j \in R$,
        \item $\theta\cdot v_i(f_i) \ge v_i(f_j)$, for any $i, j \in R$,
        \item $v_i(f_i) \ge v_i(s_j)$, for any $i\in N, j \in R$,
        \item if $v_i(f_i) < v_i(f_j)$, then $v_i(s_i) \ge v_i(s_j)$, for any $i, j \in R$.
    \end{enumerate}
\end{corollary}

\begin{proof}
    Part a) follows from \Cref{lem:NSW} since $i\not \in R$ means that $i$ is reachable from some $\theta$-envied agent (including the case where $i$ is $\theta$-envied herself). For Part b), since $i$ is reachable from some $\theta$-envied agent, if  $v_i(f_i) < v_i(f_j)$, then $j$ would also be reachable from some $\theta$-envied agent, contradicting the fact that $j\in R$. Part c) directly follows by the definition of the set $R$. For part d) notice that if it was the case that $v_i(f_i) < v_i(s_j)$, then the 
    allocation $\mathcal{F}'$ one might get from $\mathcal{F}$ by replacing $f_i$ with $s_j$
    would have strictly higher NSW than $\mathcal{F}$, i.e.,  
    $\prod_{i\in N} v_i(F'_i) > \prod_{i\in N} v_i(F_i)$, contradicting the optimality of the latter. Finally, part e) follows from the definition of the ordering $\sigma$ and the way the ``second layer'' of items is allocated: if $v_i(f_i) < v_i(f_j)$, then agent $i$ precedes agent $j$ in $\sigma$ and, thus, $s_i$ is added to $S_i$ when both $s_i$ and $s_j$ were available in $\mathcal{P}(\mathcal{S})$.
\end{proof}

At this point we are ready to state and prove the technical lemmata about the maximin share guarantees of individual agents. We start with the easy case of agents in $N\setminus R$, who only receive a single item during the first phase of the algorithm. 

\begin{lemma}[Agents in $N\setminus R$]\label{lem:ece-apx}
Fix any $\theta \ge 1$ and let $\mathcal{A}$ be the output of Match-Draft-and-Eliminate$(N, M, \theta)$. Then, for any $i\in N\setminus R$ and for any $N' \subseteq N$ with $i\in N'$, it holds that \vspace{-2ex}
    \[
    v_i(A_i) \ge \frac{\theta}{\theta + 1 - {1}/{n^\prime}} \cdot \mms_i(n^\prime, M^\prime)\,,
    \]
where $n' = |N'|$ and $M' = \bigcup_{j\in N'} A_j$.
\end{lemma}

\begin{proof}
Let $N'$ be a set as in the statement. 
We first claim that it is without loss of generality to assume that $|A_j|\ge 2$ for all $j\in N'\setminus \{i\}$.
Indeed, if $|A_j|=1$ for some $j\in N'\setminus \{i\}$, then the restriction of $\mathcal{A}$ on $N'\setminus \{j\}$
is an allocation of $M'\setminus A_j$ to $N'\setminus \{j\}$ and, 
by \Cref{lem:monotonicity}, $\mms_i(n'-1, M'\setminus A_j) \geq \mms_i(n', M')$. Thus, any guarantee shown with respect to $N'\setminus \{j\}$ directly holds for $N'$ as well. As a result, in what follows we assume that $|A_j|\ge 2$ for all $j\in N'\setminus \{i\}$. 

Let $A_{\ell}$, $\ell \in N'$, be some bundle envied by agent $i$. If no items were added to $A_{\ell}$ during the Envy-Cycle-Elimination step then we can write $A_{\ell} = \{f_k, s_k\}$, where $k\in R$ is the agent to whom the initial version of this bundle originally belonged. By part b) of \Cref{cor:sanpoulia}, $v_i(f_i) \ge v_i(f_k)$ and, hence agent $i$ envies $A_{\ell}$ only after the addition of its second, and last, item, $s_k$.
The same property holds when items were added to $A_{\ell}$ during the Envy-Cycle-Elimination step in \hyperref[line:ece]{line \ref{line:ece}}; by the definition of \Cref{alg:ece}, agent $i$ envies $A_{\ell}$ only after the addition of its last item. In  any case, call $g$ the last item added to $A_{\ell}$. By the discussion above, $v_i(A_i) \ge v_i(A_{\ell} \setminus \{g\})$.

 By part a) of \Cref{cor:sanpoulia}, we also have that $v_i(g)\le v_i(f_i)/ \theta$, as $g$ was unallocated in $\mathcal{F}$, and thus $v_i(g)\le v_i(A_i)/ \theta$, by the monotonicity of \Cref{alg:ece} (\Cref{thm:ece}). Combining these, we get
    \begin{equation*}\label{eq:3}
        v_i(A_{\ell}) \le \frac{\theta+1}{\theta} \cdot v_i(A_i).
    \end{equation*}
Given that the same bound trivially holds for non-envied agents, we can now bound the  maximin share $\mms_i(n', M')$ via the total value of $M^\prime$ for agent $i$:
        \begin{align*}
        \mms_i(n', M') & \le \frac{1}{n^\prime}v_i(M^\prime) = \frac{1}{n^\prime} \bigg[ v_i(A_i) + \sum_{k \in N^\prime \setminus \{i\}}  v_i(A_k)\bigg] \\
        &\le \frac{1}{n^\prime} \left[1 + (n^\prime - 1) \, \frac{\theta+1}{\theta}\right] v_i(A_i)  = \frac{n^\prime\theta + n^\prime - 1}{n'\theta} \, v_i(A_i).
    \end{align*}
    Therefore,
    \[
    v_i(A_i) \ge \frac{\theta}{\theta + 1 - {1}/{n^\prime}} \cdot \mms_i(n', M') \,. \qedhere
    \]
\end{proof}

We next move on to the much more challenging case of agents in $R$. Before formally proving \Cref{lem:other-apx}, we sketch the argument's structure (where $N'$ and $n'$ are as describes in the statement of the lemma). As a proxy for the maximin share $\mms_i(n',M')$ we use agent $i$'s proportional share $\frac{1}{n'}v_i(M')$; the difficulty is that, on its own, this proxy is too weak to yield anything beyond $1/2$. We first classify every other agent $\ell\in N'$ according to how her bundle received its \emph{last} item, since this controls how large $i$ may value it: bundles completed during Envy-Cycle-Elimination ($N_1$) are worth at most $i$'s own bundle plus one ``small'' item; bundles whose second (and last) item was drafted \emph{after} $s_i$ ($N_2$) are worth at most $(\theta+\delta)x$; and bundles whose second (and last) item was drafted \emph{before} $s_i$ ($N_3$) may be worth up to $2x$. The $N_3$ bundles are the heavy ones: on their own they would cap the guarantee at $1/2$, so we try to thin them out. Using the monotonicity of the maximin share (\Cref{lem:monotonicity}), we repeatedly delete one agent together with a pair of \emph{problematic} items (those forming $N_3$ bundles, together with the first item of each $N_2$ bundle); each deletion can only increase $\mms_i$, and the process terminates at a reduced instance in which problematic items no longer outnumber the agent. In this reduced instance the heavy bundles are outnumbered by the light ones (see also \Cref{fig:reduced_alloc}), and the proportional-share bound suffices for obtaining the $1/\theta$ guarantee. 

\begin{lemma}[Agents in $R$]\label{lem:other-apx}
Fix any  $\theta \ge \frac{3}{2}$ and let $\mathcal{A}$ be the output of Match-Draft-and-Eliminate$(N, M, \theta)$. Then, for any $i\in R$ and for any $N' \subseteq N$ with $i\in N'$, it holds that \vspace{-1ex}
    \[
    v_i(A_i) \ge \frac{1}{\theta} \cdot \mms_i(n^\prime, M^\prime) \,,
    \]
where $n' = |N'|$ and $M' = \bigcup_{j\in N'} A_j$.
\end{lemma}

\begin{proof}
We first consider the intermediate allocation $\mathcal{S}$, i.e., before the execution of Envy-Cycle-Elimination. For brevity, we set $v_i(f_i) = x$; then $v_i(s_i) = \delta \cdot v_i(f_i) = \delta x$, for some $\delta \in [0, 1]$, by part d) of \Cref{cor:sanpoulia}. 

Now, for agent $i$'s final bundle, $A_i$, we know that  $v_i(A_i) \ge v_i(\{f_i, s_i\})$, by \Cref{thm:ece}.

Thus
\begin{equation}\label{eq:A_to_x}
    v_i(A_i) \ge (1 + \delta) x.
\end{equation} 

Let $N'$ be a set as in the statement. Without loss of generality---as we showed in the beginning of the proof of \Cref{lem:ece-apx}---we may assume that $|A_j|\ge 2$ for all $j\in N'\setminus \{i\}$.
Also, before we go any further, notice that if $v_i(A_i) = 0$, then there are less than $n$ positively valued items for agent $i$ and are all allocated in $\mathcal{F}$. Thus, $\mms_i(n^\prime, M^\prime) = 0$ and the statement trivially holds for agent $i$. So, we assume that $v_i(A_i) > 0$.

Next, we will argue about the value agent $i$ has for any other bundle $A_{\ell}$, $\ell \in N'$, in the final allocation. We categorize the other agents of $N'$ into three groups $N_1, N_2$ and $N_3$ based on how their bundles received their last item and, thus, how agent $i$'s envy towards them may have formed. 
The set $N_1$ contains all agents of $N'\setminus\{i\}$ whose bundle received an item for the last time during the Envy-Cycle-Elimination step in \hyperref[line:ece]{line \ref{line:ece}}.
This means that the sets $N_2$ and $N_3$ contain the agents of $N'\setminus\{i\}$ whose bundle received a second item in \hyperref[line:4]{lines \ref{line:4}}\hyperref[line:6]{-\ref{line:6}} and was never augmented by Envy-Cycle-Elimination. In particular, if a bundle received its second and last item \emph{after} agent $i$ received $s_i$, then its final owner is in $N_2$, whereas if a bundle received its second and last item \emph{before} agent $i$ received $s_i$, then its final owner is in $N_3$. 

We argue about agents of different types separately. In general, although we want to refer to a bundle $A_{\ell}$ owned by $\ell \in N'$ in the final allocation, when $A_{\ell}$ received its last item (which determines whether $\ell$ is in $N_1, N_2$ or $N_3$) it might have belonged to another agent $k\in N$. To avoid cumbersome notation, as we examine the different cases for such an agent $\ell\in N'$ below, we will consistently assume that $A_{\ell}$ belonged to agent $k\in N$ when it received its last item and use $\mathcal{B} = (B_1,\ldots,B_n)$ to denote the allocation at the time. 

\begin{itemize}[leftmargin=*]
    \item $\ell\in N_1$. That is, $A_{\ell} = B_k$ and agent $k$ received the last item, say $g$, in \hyperref[line:ece]{line \ref{line:ece}}. 
    Then we have that $v_i(A_{\ell}) = v_i(B_{k})  \le v_i(B_i) + v_i(g) \le v_i(A_i) + \delta x$.
    To see this, at the time agent $i$ was not envious of agent $k$ before item $g$ was added (as $k$ must have been a source of the envy graph), i.e., $v_i(B_i) \ge v_i(B_{k}\setminus\{g\})$. Further, we know that $v_i(g) \le v_i(s_i) = \delta x$ (or agent $i$ would have picked $g$ instead of $s_i$ in \hyperref[line:4]{lines \ref{line:4}}\hyperref[line:6]{-\ref{line:6}}).  Using \eqref{eq:A_to_x} we get that 
    \begin{equation}\label{eq:n1bound}
    v_i(A_k) \le \frac{1 + 2\delta}{1 + \delta} v_i(A_i).
    \end{equation}
    
    \item $\ell \in N_2$. That is, $A_{\ell} = B_k = \{f_k, s_k\}$ and $s_k$ was allocated after $s_i$.
    Then we have that $\theta \cdot v_i(f_i) \ge v_i(f_k)$ by part c) of \Cref{cor:sanpoulia} and $v_i(s_k) \le v_i(s_i) = \delta x$ (or agent $i$ would have picked $s_k$ instead of $s_i$ in \hyperref[line:4]{lines \ref{line:4}}\hyperref[line:6]{-\ref{line:6}}).
    Thus, in total $v_i(A_{\ell}) = v_i(B_{k})  \le (\theta + \delta)x$. Again, using \eqref{eq:A_to_x} we get that 
    \begin{equation}\label{eq:n2bound}
    v_i(A_k) \le \frac{\theta + \delta}{1 + \delta} v_i(A_i).
    \end{equation}

    \item $\ell \in N_3$.  That is, $A_{\ell} = B_k = \{f_k, s_k\}$ and $s_k$ was allocated before $s_i$.
    Now, we have that agent $i$ prefers neither $f_k$ (or else $i$ would be before $k$ in the topological ordering $\sigma$ in \hyperref[line:order]{line \ref{line:order}} and $s_i$ would have been allocated before $s_k$) nor $s_k$ (by part d) of \Cref{cor:sanpoulia}) to $f_i$. Thus $v_i(A_{\ell}) = v_i(B_{k}) \le 2x$ and using \eqref{eq:A_to_x} we get that 
    \begin{equation}\label{eq:n3bound}
        v_i(A_k) \le \frac{2}{1 + \delta} v_i(A_i).
    \end{equation}
    \end{itemize}

Our next step is to further reduce the instance. 
Let $n_1, n_2, n_3$ be the cardinalities of $N_1, N_2, N_3$ respectively. 

We will say that an item is \emph{\prob} if it belongs to some agent in $N_3$ or if it is the first item added (via $\mathcal{F}$) to a bundle  that belongs to some agent in $N_2$. The number of \prob items is $p:= 2n_3 + n_2$, by definition. 

We claim that as long as $p$ is larger than the total number of agents, $n' = n_1 + n_2 + n_3+ 1$, we can remove a pair of \prob items and an agent without ever reducing the maximin share of agent $i$. 
Indeed, consider a $n'$-maximin share defining partition $\mathcal{C}$ for agent $i$ with respect to $M'$. 
By the pigeonhole principle, if $p>n'$, then there is a bundle of $\mathcal{C}$ containing two \prob goods, say $g_1$, $g_2$. We remove those two goods and reduce the number of agents by one; by \Cref{lem:monotonicity}, we have $\mms_i(n'-1, M'\setminus \{g_1, g_2\}) \ge \mms_i(n', M')$.

So, we may assume that we have already performed these reductions up to the point where there are no more \prob items than agents in the reduced instance. Next, we claim that the remaining items can be partitioned in a \emph{virtual} allocation that closely resembles the restriction of $\mathcal{A}$ on $N'$ after the removal of agents with singletons, and crucially shares the same properties with respect to the bounds derived above in \eqref{eq:n1bound}, \eqref{eq:n2bound} and \eqref{eq:n3bound}. Indeed, if $\ell\in N_1 \cup \{i\}$, then $A_{\ell}$ has remained intact; if $\ell\in N_2$ and $A_{\ell}$ lost its \prob item, this can be replaced by an arbitrary \prob item originally belonging to some agent of $N_2$ or $N_3$; and, finally, if there are more \prob items, these can be arbitrarily partitioned in sets of size $2$ and be given to (some of the) agents of $N_3$. 

To avoid introducing additional notation for essentially the same thing, from this point onward, we are going to use $N', M', N_1, N_2, N_3, n', n_1, n_2, n_3$ for the final reduced instance. The virtual allocation described above will be denoted by $\mathcal{A'} = (A'_1,\ldots,A'_n)$. As we mentioned above, and easily follows by their definition, the bundles of $\mathcal{A'}$ satisfy \eqref{eq:n1bound}, \eqref{eq:n2bound} and \eqref{eq:n3bound}, whereas $A'_i = A_i$.  Further, since the instance is reduced, we now have $2n_3 + n_2 \le n_1 + n_2 + n_3 + 1$ or, equivalently, $n_3 \le n_1 + 1$.

Next, we will bound the total value of agent $i$ for $M'$ in this reduced instance. 
\begin{align*}
    v_i(M^\prime) &= \sum_{k \in N_1} v_i(A'_k) + \sum_{k \in N_2} v_i(A'_k) + \sum_{k \in N_3} v_i(A'_k) + v_i(A_i)\\
        &\le n_1 \frac{1 + 2\delta}{1 + \delta} v_i(A_i) +
        n_2 \frac{\theta + \delta}{1 + \delta} v_i(A_i) +
        n_3 \frac{2}{1 + \delta} v_i(A_i)  + v_i(A_i) .
\end{align*}
By dividing with $v_i(A_i)$ and rearranging some terms, we get
\begin{align*}
        \frac{v_i(M^\prime)}{v_i(A_i)} &\le n_1 \frac{1 + 2\delta}{1 + \delta}+
        n_2 \frac{\theta + \delta}{1 + \delta} +
        n_3 \frac{2}{1 + \delta} + 1 \\  
        &= n_1\frac{1 + 2\delta}{1 + \delta}+
        (n^\prime - n_1 - n_3 - 1) \frac{\theta + \delta}{1 + \delta} +
        n_3 \frac{2}{1 + \delta} + 1 \\
        &= n^\prime \frac{\theta + \delta}{1 + \delta} +  n_1\frac{1+\delta-\theta}{1+\delta} + n_3\frac{2 - \delta - \theta}{1+\delta} - \frac{\theta-1}{1+\delta}
\end{align*}
Now we have cases based on the sign of the coefficients of $n_1$ and $n_3$. Since $\theta \ge 3/2$ and at most one of $1+\delta$ or $2-\delta$ can be greater than $3/2$, at most one term is positive. If they are both negative, then we immediately obtain:
\begin{align*}
\frac{v_i(M^\prime)}{v_i(A_i)} \le n^\prime \frac{\theta + \delta}{1 + \delta} \implies
v_i(A_i) \ge \frac{1+\delta}{\theta+\delta}  \frac{v_i(M^\prime)}{n^\prime} \ge \frac{1}{\theta} \cdot\mms_i(n', M'),
\end{align*}
where the last inequality follows from the fact that  $\frac{1+\delta}{\theta+\delta}$ as a function of $\delta\in [0,1]$ is minimized for $\delta = 0$.

Next, if the coefficient of $n_3$ is negative, we have that
\begin{align*}
        \frac{v_i(M^\prime)}{v_i(A_i)} &\le n^\prime \frac{\theta + \delta}{1 + \delta} +  n_1\frac{1+\delta-\theta}{1+\delta} \\
        &\le n^\prime \frac{\theta + \delta}{1 + \delta} +  n^\prime\frac{1+\delta-\theta}{1+\delta} = n^\prime \frac{1+2\delta}{1+\delta} \implies \\
        v_i(A_i) &\ge \frac{1+\delta}{1+2\delta}\cdot \frac{v_i(M^\prime)}{n^\prime}  \ge \frac{2}{3} \cdot\mms_i(n', M') \ge \frac{1}{\theta} \cdot \mms_i(n', M'),
\end{align*}
where the second to last inequality follows from the fact that  $\frac{1+\delta}{1+2\delta}$ as a function of $\delta\in [0,1]$ is minimized for $\delta = 1$.

Finally, if the coefficient of $n_3$ is positive, we have that
\begin{align*}
        \frac{v_i(M^\prime)}{v_i(A_i)} &\le n^\prime \frac{\theta + \delta}{1 + \delta} +  n_1\frac{1+\delta-\theta}{1+\delta} + (n_1 + 1)\frac{2 - \delta - \theta}{1+\delta} - \frac{\theta-1}{1+\delta} \\
        &= n^\prime \frac{\theta + \delta}{1 + \delta} +  n_1\frac{3-2\theta}{1+\delta} + \frac{3 - 2\theta -\delta}{1+\delta} \le n^\prime \frac{\theta + \delta}{1 + \delta} \implies \\
        v_i(A_i) &\ge \frac{1+\delta}{\theta+\delta} \cdot\frac{v_i(M^\prime)}{n^\prime} \ge \frac{1}{\theta} \cdot\mms_i(n', M').
\end{align*}

Since both reductions (removing a bundle of cardinality $1$ along with the corresponding agent and removing $2$ \prob items along with an agent in $N_2$) weakly increase the maximin share of agent $i$, our $1/\theta$ guarantee holds for any arbitrary $N'\subseteq N$.
\end{proof}

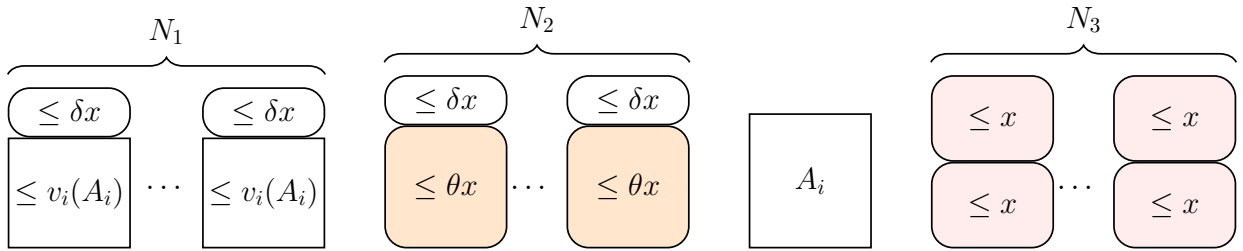
\begin{figure}[ht]
    \centering
    \resizebox{1\textwidth}{!}{%
    \begin{tikzpicture}[
        box/.style={
            draw, 
            thick, 
            anchor=south, 
            inner sep=0pt,
            minimum width=2.0cm 
        },
        v_large/.style={box, minimum height=2.6cm, font=\Large},
        large/.style={box, minimum height=2.0cm, font=\Large},
        med/.style={box, minimum height=1.4cm, font=\Large},
        aone/.style={box, minimum height=2.2cm, font=\Large},
        new/.style={box, minimum height=1.8cm, font=\Large},
        small/.style={box, minimum height=0.8cm, font=\Large},
        dots/.style={anchor=center, font=\Large}
    ]

        \node[new] (N1) at (6.2,0) {$\le v_i(A_i)$};
        \node[dots] at (7.8, 1.0) {$\cdots$}; 
        \node[new] (N2) at (9.4, 0) {$\le v_i(A_i)$};

        \node[large,rounded corners=10pt,fill=orange!20] (M1) at (12.4 ,0) {$\le \theta x$};
        \node[dots] at (13.8, 1.0) {$\cdots$}; 
        \node[large, rounded corners=10pt, fill=orange!20] (M2) at (15.4, 0) {$\le \theta x$};

        \node[aone] (X) at (18.4,0) {$A_i$};

        \node[med, fill=pink!30, rounded corners=10pt] (S1) at (21.4,0) {$\le x$};
        \node[dots] at (22.8, 1.0) {$\cdots$}; 
        \node[med, fill=pink!30, rounded corners=10pt] (S2) at (24.4,0) {$\le x$};

        \node[small, rounded corners=10pt] (Q1) at (N1.north) {$\le \delta x$};
        \node[small, rounded corners=10pt] (Q2) at (N2.north) {$\le \delta x$};
        
        \node[small, rounded corners=10pt] (T1) at (M1.north) {$\le \delta x$};
        \node[small, rounded corners=10pt] (T2) at (M2.north) {$\le \delta x$};

        \node[med, fill=pink!30, rounded corners=10pt] (TS1) at (S1.north) {$\le x$};
        \node[med, fill=pink!30, rounded corners=10pt] (TS2) at (S2.north) {$\le x$};

        \tikzset{
            brace style/.style={
                decorate, 
                decoration={brace, amplitude=8pt, raise=6pt}, 
                thick
            }
        }

        \draw [brace style] (Q1.north west) -- (Q2.north east) 
            node [midway, above=16pt] {\Large$N_1$};

        \draw [brace style] (T1.north west) -- (T2.north east) 
            node [midway, above=16pt] {\Large$N_2$};

        \draw [brace style] (TS1.north west) -- (TS2.north east) 
            node [midway, above=16pt] {\Large$N_3$};

    \end{tikzpicture}%
    }
    \caption{The structure of the allocation in the proof of \Cref{lem:other-apx} after the reductions that remove agents with singleton bundles, where square boxes denote bundles while rounded ones correspond to a single item. This, however, is also the structure of the virtual allocation $\mathcal{A'}$ we reconstruct after the reductions that remove agents along with pairs of \prob items. The latter are colored; orange for $N_2$ and pink for $N_3$.}
    \label{fig:reduced_alloc}
\end{figure}

\begin{proof}[Proof of \Cref{thm:main}]
First notice that every step of the algorithm can be computed in polynomial time. Indeed, finding a NSW-maximizing matching \citep{Farhadietal21}, constructing $G_\mathcal{F}$ and computing $R$, finding the topological ordering $\sigma$, allocating the second layer of items, and running Envy-Cycle-Elimination (\Cref{thm:xece}) can all be done efficiently. Thus, Match-Draft-and-Eliminate$(N, M, \phi)$ returns $\mathcal{A}$ in polynomial time. 

Next, notice that, for $\theta\ge 1$, by \Cref{cor:sanpoulia},  the allocation $\mathcal{S}$ completed in \hyperref[line:8]{line \ref{line:8}} of \Cref{alg:main} is \efo. Then, part a) of \Cref{thm:xece} guarantees that $\mathcal{A}$ is \efo as well.

Similarly, for $\theta\in [1,2]$, again by \Cref{cor:sanpoulia}, the allocation $\mathcal{S}$ is $1/\theta$-\efx and all agents are $\theta$-content. Then, Theorem 3 of \cite{MS23} guarantees that $\mathcal{A}$ is $\min(\frac{1}{\theta}, \frac{\theta}{\theta + 1})$-\efx. For $\theta = \phi = \frac{\sqrt{5}+1}{2}$, however, we have $\min(\frac{1}{\theta}, \frac{\theta}{\theta + 1}) = \phi - 1$.

Finally, we invoke \Cref{lem:ece-apx} for the agents of $N\setminus R$ and \Cref{lem:other-apx} for the agents of $R$ with parameter $\theta = \phi$.
    Thus, the approximation factor of the final allocation with respect to \gmms is 
    \[
    \min \left\{ \frac{\phi}{\phi + 1}, \frac{1}{\phi} \right\}. 
    \]
    By the definition of $\phi$ (as the positive solution of $x^2-x-1=0$), the two quantities are equal to each other and also equal to $\phi - 1$.
\end{proof}

We also present a version of the \Cref{th:main} that results in an improved factor for any fixed $n$, at the expense of the \efx approximation.

\begin{theorem}\label{th:main}
    Set $\theta_n = \max\Big\{\frac{3}{2} , \frac{1 + \sqrt{5 - {4}/{n}}}{2}\Big\}$. Then, the allocation $\mathcal{A}$ returned by Match-Draft-and-Eliminate$(N, M, \theta_n)$ is $\min\left\{\frac{2}{3}, \frac{\sqrt{5n^2 - 4n} - n}{2n-2}\right\}$-\gmms.
\end{theorem}

\begin{proof}
    We invoke \Cref{lem:ece-apx} for the agents of $N\setminus R$ and \Cref{lem:other-apx} for the agents of $R$ with parameter $\theta_n \ge {3}/{2}$.
    Thus, the approximation factor of the final allocation with respect to \gmms is 
    \[
    \min \left\{ \frac{\theta_n}{\theta_n + 1 - \frac{1}{n}}, \frac{1}{\theta_n} \right\}. 
    \]
    Equating the two quantities gives
    \[
    \theta_n^2 - \theta_n - 1 + \frac{1}{n} = 0 \implies \theta_n = \frac{1 \pm \sqrt{5 - \nicefrac{4}{n}}}{2}
    \]
    Of course, to satisfy the condition $\theta_n \ge {3}/{2}$ we keep the larger solution. Moreover, note that for $n \le 4$, we have $\frac{1 + \sqrt{5 - {4}/{n}}}{2} \le  {3}/{2}$, thus we must set $\theta_n = {3}/{2}$ to obtain the approximation factor ${2}/{3}$. For $n \ge 5$, by setting $\theta_n = \frac{1 + \sqrt{5 - {4}/{n}}}{2}$ we obtain the factor $\frac{\sqrt{5n^2 - 4n} - n}{2n-2}$. 
\end{proof}\vspace{-3ex}

\begin{figure}[!ht]
    \centering
    \resizebox{0.5\textwidth}{!}{%
\begin{tikzpicture}
\def\N{30}

\begin{axis}[
    xlabel={$n$},
    ylabel={},
    xmin=2, xmax=\N+1,
    ymin=0.615, ymax=0.675,
    grid=both,
    major grid style={line width=.2pt,draw=gray!50},
    legend pos=north east,
    ytick={0.62, 0.63, 0.64, 0.65, 0.66, 0.67},
    yticklabel style={
        /pgf/number format/fixed,
        /pgf/number format/precision=2
    }
]

\addplot[
    only marks,
    mark=*,
    color=blue,
    samples at={3,4,...,\N} 
] { x < 5 ? 2/3 : 2 / (1 + sqrt(5 - 4/x)) };
\addlegendentry{\Cref{th:main}}

\addplot[
    domain=2:\N+1,
    dashed,
    color=red,
    thick,
    samples=2
] { (sqrt(5) - 1) / 2 }; 
\addlegendentry{$\phi - 1 \approx 0.618$}

\end{axis}
\end{tikzpicture}%
    }
    \caption{The approximation guarantee of \Cref{alg:main} and a comparison with $\phi-1$ for relatively small values of $n$; note that the approximation ratio is $2/3$ for $n\in\{3, 4\}$.}
    \label{fig:placeholder}
\end{figure}
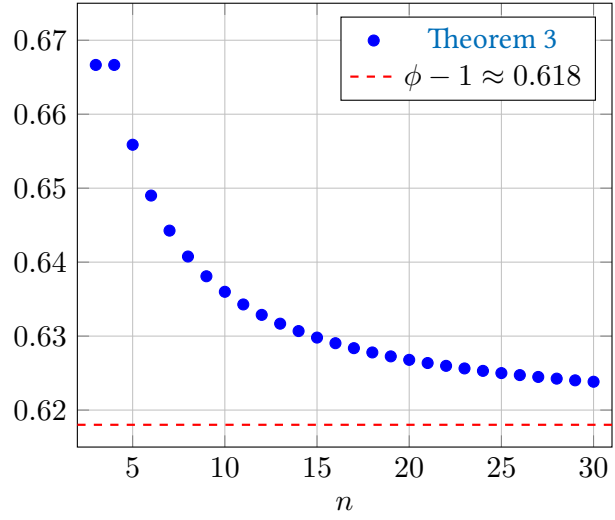

\vspace{.5cm}
Using the proportional share, i.e., the $1/n$-th of the total value, as a proxy for maximin share is very convenient but, more often than not, results in rather loose bounds. Although our analysis leading to $\phi-1$ is very detailed, it is still  surprising that it is tight, as shown below. The proof of \Cref{th:upper} can easily be adjusted to show that corresponding algorithms of \cite{amanatidis2020multiple} and  \cite{Farhadietal21} cannot go beyond an approximation factor of $\phi - 1$ either.

\begin{theorem}\label{th:upper}
 Let $\delta >0$ and $\theta \ge 1$. \Cref{alg:main} with parameter $\theta$ does not always produce $(\phi-1+\delta)$-\gmms allocations. 
\end{theorem}

\begin{proof}[Proof of \Cref{th:upper}]
    We begin by describing a family of instances that cause the algorithm to fail to produce a $(\phi-1+\delta)$-\mmsT allocation for any positive constant $\delta$, thus implying the corresponding result for \gmms as well. In the following table we present the agents' values for the $n$ most preferred goods of agents 2 to $n$. 
    Moreover, if $v_1(g_k) \le 1$ and $v_i(g_k) < 1$ for all $i>1$ and $k > n$,  then the unique NSW-maximizing matching allocation is the one that sets $f_i = g_i$.
    \[
    \begin{array}{ccccccc}
    \toprule
    & g_1  & g_2    & g_3    & \cdots & g_{n-1} & g_n \\
    \midrule
    a_1     & \phi + \epsilon & 0      & 0      & \cdots & 0       & 0 \\
    a_2     & \phi + \epsilon & \phi + \nicefrac{\epsilon}{2}   & \phi      & \cdots & \phi      & 1 \\
    a_3     & \phi + \epsilon & \phi   & \phi + \nicefrac{\epsilon}{2}  & \cdots & \phi       & 1 \\
    \vdots  & \vdots          & \vdots & \vdots & \ddots & \vdots  & \vdots \\
    a_{n-1} & \phi + \epsilon & \phi   & \phi   & \cdots & \phi + \nicefrac{\epsilon}{2}    & 1 \\
    a_n     & \phi & \phi  & \phi  & \cdots & \phi & 1  \\
    \bottomrule
    \end{array}
    \]\smallskip

    Next, we have two cases based on whether  we run \Cref{alg:main} with $\theta > \phi$ or not. \\[5pt]
    \noindent\underline{Case 1: $\theta > \phi$.} We focus on agent $n$. It holds that $v_n(g_k) = \frac{\phi-1}{n-1}$ for $k=n+1,\dots, 2n-1$ and $v_n(g_k) = 0$ otherwise.  Observe that agent $n$ creates her unique $n$-maximin share defining partition by bundling $g_n,\ldots,g_{2n-1}$ together, resulting in $\mms_n(n, M) = \phi$. However, agent $n$ does not $\theta$-envy agents $2$ to $n-1$ (nor are they $\theta$-envied by anyone else), thus they will all be placed in $R$ and receive a second item. By having agents $2$ to $n-1$ prefer the items in $\{g_{n+1}, \dots, g_{2n-1}\}$ we will end with a final allocation where $v_n(A_n) = 1 + \frac{\phi}{n-1}$. Hence the approximation ratio is at most \[
    \frac{v_n(A_n)}{\mms_n(n, M)} = \frac{1 + \frac{\phi}{n-1}}{\phi }
    \]
    As $n$ grows this ratio becomes arbitrarily close to ${1}/{\phi} = \phi - 1$, so it cannot be at least $\phi-1+\delta$ for sufficiently large $n$. \smallskip

    \noindent\underline{Case 2: $\theta \le \phi$.} Now only agent $n$ will be in $R$ and she will receive some item from $\{g_{n+1}, \dots, g_{2n-1}\}$; without loss of generality, let it be $g_{n+1}$. Next, we impose a moderate constraint on the valuations of agents 2 to $n-1$: their total value  for all the items except their $n$ most preferred ones shown above is strictly less than ${\epsilon}/{2}$. Therefore, they are always envious of agent 1. In other words, agent 1 will receive another good only if she participates in a cycle elimination and exchanges her bundle. Now, we carefully define her valuation function for the remaining items:
    \[
     v_1(g_k) = \begin{cases}
        1, &k - n \text{ is a multiple of } \frac{m-n}{n-1}\\
        c, &\text{otherwise}
     \end{cases}
    \]
    Now, let us pick $c = \frac{(n-1)\phi}{m+1-2n}$. This makes the total value of those items equal to $(n-1)\phi$. Next, we pick $m$ such that $n-1$ divides $m+1-2n$. This allows us to create $n-1$ identical for agent 1 bundles $B_2, \dots, B_n$  of those small items with value exactly $\phi$ each. Note that we may construct the sets such that $g_i \in B_i$ for all $i \ge 2$ and, moreover, such that $g_{n+1} \in B_n$.
    
    Assume, for now, that the allocation $\left(\{g_1\}, B_2, \dots, B_n \right)$ can be produced by the execution of the Envy-Cycle-Elimination algorithm starting from the allocation $\left(\{g_1\}, \dots, \{g_{n-1}\}, \{g_n, g_{n+1}\}\right)$. Then, the remaining items are precisely those that agent 1 values at $1$, let us call them $h$. Note that there are exactly $n-1$ copy of $h$, with the last one occurring for $k = m$. Each other agent will receive, then, exactly one of these, since agent 1 becomes envious of the receiving agent once this happens. But there are enough agents to allocate all the remaining items and, thus, in the final allocation we will have $A_1 = \{g_1\}$ and $v_1(A_1) = \phi + \epsilon$.
    
    Next, we want to bound agent 1's maximin share. To that end, note that agent 1 can pair each $g_k$ she values at 1 with a bundle of $c$-valued items, let us call it $C_k$, such that $v_1(\{g_k\} \cup C_k) \approx \phi$. Or, equivalently, $v_1(C_k) \approx \phi - 1$. This way she uses  value approximately equal to $(n-1)(\phi-1)$ out of the total  value of the small items (which is $(n-1)\phi$). The remaining value, which is about $n-1$ will be equally distributed among all bundles. Hence, we get that $\mms_1(n, M) \ge \phi + \frac{n-1}{n} - c$, where the extra $c$ amounts to the discretization error treating the small items as a continuous mass. Still, we can select $m$ sufficiently large such that this error becomes negligible, or in other words $\mms_1(n, M) = \phi + \frac{n-1}{n} - \epsilon'$ for small $\epsilon' > 0$. 

    Putting everything together results in 
    \[
    \frac{v_1(A_1)}{\mms_1(n, M)} \le \frac{\phi + \epsilon}{\phi + \frac{n-1}{n} - \epsilon'}
    \]
    As $n$ grows and for sufficiently small $\epsilon, \epsilon'$ this ratio becomes arbitrarily close to ${\phi}/{\phi + 1} = \phi - 1$, so it cannot be at least $\phi-1+\delta$ for sufficiently large $n$.
    
    To complete the proof it remains to show that the allocation $(\{g_1\}, B_2 \cup \{h\},\allowbreak \dots, \allowbreak B_n \cup \{h\})$ can be produced via the Envy-Cycle-Elimination algorithm. Note that if agent $n$ receives all those items, then there is no envy between agents in 2 to $n$. Therefore, each one of them can receive items until agent 1 starts envying them. Here, $h$ items come into play. Notice that an $h$ item is positioned after  
    exactly $\frac{m+1-2n}{n-1}$ many $c$-valued items (i.e., the $h$ items occur every $\frac{m+1-2n}{n-1}$ positions among the $c$-valued items).
    In other words, agent $n$ receives $B_n$ and then an item $h$. At this point, agent 1 envies her and agent $n-1$ is the new source. She receives $B_{n-1}$ and then an $h$ and the pattern continues until each agent has exactly one $B$ set and one $h$ item.
\end{proof}

\section{Improved Guarantees for Restricted Settings}\label{sec:restricted}

In this section we exhibit some interesting special cases where we can have even better approximation guarantees for \gmms.

\subsection{Agreement on the $n$ most desirable items}\label{sec:topn}

A family of instances that has recently attracted more attention in the literature is the class we refer to as top-$n$ instances. In a top-$n$ instance, all agents agree on the set of their $n$ most desirable items (but not necessarily on their value or on their ranking of these items). The main motivation for examining  this class is that it is a relaxation of ordered valuations (where the agents agree on the ranking of all items), which has been extensively studied in the fair division literature. Additionally, recent works have established improved approximation guarantees for top-$n$ instances for various fairness notions, see e.g., \citep{MS23,ChristoforidisS24,MMP25,AR26}. 

\cite{MS23} proposed an algorithm that achieves a $2/3$-\efx guarantee for top-$n$ instances, improving upon the best-known $(\phi-1)$ approximation of \efx for additive instances. Our main result in this subsection is that we can have the same improvement for \gmms as well. This is achieved by a slight modification of the algorithm in \cite{MS23}, and its pseudocode is presented as \Cref{alg:topn}. The main difference  is  the condition in \hyperref[line:topn-if]{line \ref{line:topn-if}}. This is similar in spirit but simpler than \Cref{alg:main}, in the sense that again each agent receives one or two items before Envy-Cycle-Elimination begins.

\begin{algorithm}[ht]
	\DontPrintSemicolon 
		Let $T$ be the set of the $n$ most desirable items and set $B = M \setminus T$\;
        \For{each agent $i = 1, \dots, n$ }{
            $f^+_i = \argmax_{g \in T} v_i(g)$\;
            $f^-_i = \argmin_{g \in T} v_i(g)$\;
            $s_i = \argmax_{g \in B} v_i(g)$\;
            \If{$v_i(f^-_i) + v_i(s_i) \ge v_i(f^+_i)$ \label{line:topn-if}}{
                $A_i = \{s_i\}$\;
                $B = B \setminus \{s_i\}$
            }
            \Else{
                $A_i = \{f^+_i\}$ \label{line:ffirst}\;
                $T = T \setminus \{f^+_i\}$
            }
        }
        \For{each agent $i = n, \dots, 1$}{
            \If{$A_i = \{s_i\}$}{
                $f_i = \argmax_{g \in T} v_i(g)$\;
                $A_i = A_i \cup \{f_i\}\label{line:fsecond}$
            }
        }
        $\mathcal{A} = \text{Envy-Cycle-Elimination}(N, \mathcal{A}, B)$\;
		\Return $\mathcal{A}$ \;
        
	\caption{Select-and-Eliminate$(N, M)$} \label{alg:topn}
\end{algorithm}

\begin{lemma}\label{lem:2content}
Let $T$ be the set of the $n$ most desirable items. Then, all agents are 2-content after all items of $T$ get allocated.
\end{lemma}
\begin{proof}
    First, observe that all items from $T$ are eventually allocated after the execution of the second for-loop, and each agent receives exactly one from that set. Let $\mathcal{A}$ be the allocation produced right after the second for-loop. Now, we have two cases based on whether agent $i$ first receives $s_i$ (through the if condition) or $f_i^+$. If she first receives $s_i$ then she also receives in \hyperref[line:fsecond]{line \ref{line:fsecond}} some item $f_i$, which she values at least as much as $f^-_i$. Then $
    v_i(A_i) \ge v_i(\{f_i^-, s_i\}) \ge 2\cdot v_i(s_i) \ge 2\cdot v_i(g)\ \forall g \in \pool(\mathcal{A}). $
    Similarly, if she receives $f_i^+$ in \hyperref[line:ffirst]{line \ref{line:ffirst}}, then $v_i(A_i) = v_i(f^+_i) > v_i(\{f_i^-, s_i\})$, due to \hyperref[line:topn-if]{line \ref{line:topn-if}}. We then conclude as in the former case.
\end{proof}

\begin{theorem}\label{thm:topn}
    In top-$n$ instances, Select-and-Eliminate$(N, M)$ returns in polynomial time an allocation $\mathcal{A}$ which is $2/3$-\gmms, \efo, and $2/3$-\efx.
\end{theorem}

\begin{proof}
To begin with, we fix an agent $i$ and a subset of agents $N^\prime\ni i$. If agent $i$ received her first item in \hyperref[line:fsecond]{line \ref{line:fsecond}} then we bound her value via the following claim.

\begin{claim}\label{claim:lemma}
    Let $\mathcal{A}$ be the output of Select-and-Eliminate$(N, M)$. Then, for any agent $i$ who receives a good in \hyperref[line:ffirst]{line \ref{line:ffirst}} and for any $N' \subseteq N$ with $i\in N'$, it holds that
    \[
    v_i(A_i) \ge \frac{2}{3 - {1}/{n^\prime}} \cdot \mms_i(n^\prime, M^\prime)\,,
    \]
where $n' = |N'|$ and $M' = \bigcup_{j\in N'} A_j$.
\end{claim}

\begin{proof}[Proof of \Cref{claim:lemma}.]
    The proof of the Claim follows almost verbatim the proof of \Cref{lem:ece-apx} for $\theta = 2$. The only differences are that the analog of part a) of \Cref{cor:sanpoulia} for $\theta = 2$ is now a direct consequence of \Cref{lem:2content}, whereas the analog of part b) of \Cref{cor:sanpoulia} now follows by the order in which the goods are allocated (\hyperref[line:ffirst]{line \ref{line:ffirst}} versus \hyperref[line:fsecond]{line \ref{line:fsecond}}).
\end{proof}

We now focus on the case where agent $i$ receives a good (in fact, her second item) in \hyperref[line:fsecond]{line \ref{line:fsecond}}.
We set $v_i(f_i^-) = x$ and $v_i(s_i) = \delta x$ for $\delta \in [0,1]$. By the condition in 
\hyperref[line:topn-if]{line \ref{line:topn-if}},
we get that for any agent $j$ who received her second item in \hyperref[line:fsecond]{line \ref{line:fsecond}} it holds that $v_i(f_j) \le (1+\delta)x$. 
Next, we define two groups of agents, $N_{12}$ and $N_3$, where $N_3$ contains all agents of $N'$ who received a second item in \hyperref[line:fsecond]{line \ref{line:fsecond}} \textit{after} agent $i$ received $f_i$ and never received another good; $N_{12}$ contains all the remaining agents except $i$, i.e., $N_{12} = N' \setminus (N_3 \cup \{i\})$.\footnote{We use this naming convention to draw a parallel with \Cref{lem:other-apx}: $N_3$ corresponds to the group of agents with the same name therein but in the top-$n$ instance we may merge the analogs of groups $N_1$ and $N_2$.} Here, \textit{problematic} items are only those belonging to bundles of agents in $N_3$. 
We reduce the instance in a similar fashion to \Cref{lem:other-apx}, using \Cref{lem:monotonicity}, to obtain a virtual allocation that has a same structure to the restriction of $\mathcal{A}$ on $N'$ after the removal of agents with singletons (see also \Cref{fig:reduced_alloc2}). Again, to avoid introducing additional notation, in what follows we use $N', M', N_{12},  N_3$ for the final reduced instance. For the cardinalities of $N_{12}$ and $N_3$, $n_{12}$ and $n_3$  respectively, it holds $n_3 \le n_{12} + 1$, by \Cref{lem:monotonicity} and the pigeonhole principle. 

Working like we did in the proof of \Cref{th:main} to obtain \eqref{eq:n1bound} (for part of agents in $N_{12}$), \eqref{eq:n2bound} (for the remaining of agents in $N_{12}$; only here instead of $\theta$ we have $1+\delta$ by \hyperref[line:topn-if]{line \ref{line:topn-if}} as we mentioned above), and  \eqref{eq:n3bound} (for agents in $N_{3}$; with the additional observation that $\frac{v_i(f_i) + x}{v_i(f_i) + \delta x} \le \frac{2}{1+\delta}$, since $v_i(f_i)\ge x$), we get the inequality:
\begin{align*}
        \frac{v_i(M^\prime)}{v_i(A_i)} &\le n_{12} \frac{1 + 2\delta}{1 + \delta}+
        n_3 \frac{2}{1 + \delta} + 1 \\
        &= (n - 1 - n_3) \frac{1 + 2\delta}{1 + \delta} + n_3\frac{2}{1+\delta} + 1\\
        &= n \frac{1 + 2\delta}{1 + \delta} - \frac{\delta}{1+\delta} + n_3 \frac{1 - 2\delta}{1 + \delta}\,.
        \end{align*}
For $\delta \ge \frac{1}{2}$ we directly obtain the bound
\[
    \frac{v_i(M^\prime)}{v_i(A_i)} \le n \frac{1 + 2\delta}{1 + \delta} \implies v_i(A_i) \ge \frac{1+\delta}{1+2\delta} \cdot\frac{v_i(M^\prime)}{n} \ge \frac{2}{3}\cdot \frac{v_i(M^\prime)}{n}\,.
\]
for any $\delta \in [\frac{1}{2}, 1]$.
For $\delta < \frac{1}{2}$ the coefficient of $n_3$ is positive. However, combining the relations $n_3 \le n_{12} + 1$ and $n_{12} +  n_3 + 1 = n$ we can upper bound $n_3$ by $\frac{n}{2}$. Hence
\begin{align*}
    \frac{v_i(M^\prime)}{v_i(A_i)} &< n \frac{1 + 2\delta}{1 + \delta} + \frac{n}{2} \cdot \frac{1 - 2\delta}{1 + \delta} = n \frac{3 + 2\delta}{2 + 2\delta} \implies \\
    v_i(A_i) &> \frac{2 + 2\delta}{3 + 2\delta}\cdot \frac{v_i(M^\prime)}{n} > \frac{2}{3} \cdot \frac{v_i(M^\prime)}{n} \,.
\end{align*}
Piecing everything together, we conclude that
\[
v_i(A_i) \ge \frac{2}{3}\cdot \frac{v_i(M^\prime)}{n} \ge \frac{2}{3}\cdot \bmu_i(n^\prime, M^\prime).
\]
This completes the \gmms part of the proof. For both \efo and \efx note that the intermediate allocation after the second for-loop is \efx, and thus also \efo. Then, the \efo property is maintained by the Envy-Cycle-Elimination algorithm. The guarantee about \efx is obtained via Theorem 3 of \cite{MS23} and \Cref{lem:2content}. Finally, note that both for-loops run in polynomial time, as does the Envy-Cycle-Elimination algorithm so Select-and-Eliminate is efficient. \end{proof}

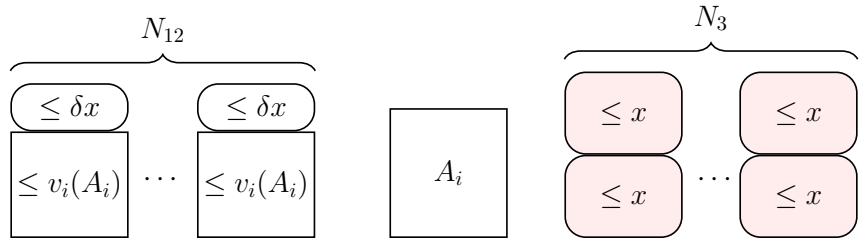
\begin{figure}[ht]
    \centering
    \resizebox{.69\textwidth}{!}{%
    \begin{tikzpicture}[
        box/.style={
            draw, 
            thick, 
            anchor=south, 
            inner sep=0pt,
            minimum width=2.0cm 
        },
        v_large/.style={box, minimum height=2.6cm, font=\Large},
        large/.style={box, minimum height=2.0cm, font=\Large},
        med/.style={box, minimum height=1.4cm, font=\Large},
        aone/.style={box, minimum height=2.2cm, font=\Large},
        new/.style={box, minimum height=1.8cm, font=\Large},
        small/.style={box, minimum height=0.8cm, font=\Large},
        dots/.style={anchor=center, font=\Large}
    ]

        \node[new] (N1) at (9.5,0) {$\le v_i(A_i)$};
        \node[dots] at (11.1, 1.0) {$\cdots$}; 
        \node[new] (N2) at (12.7, 0) {$\le v_i(A_i)$};

        \node[aone] (X) at (16,0) {$A_i$};

        \node[med, fill=pink!30, rounded corners=10pt] (S1) at (19,0) {$\le x$};
        \node[dots] at (20.6, 1.0) {$\cdots$}; 
        \node[med, fill=pink!30, rounded corners=10pt] (S2) at (22,0) {$\le x$};

        \node[small, rounded corners=10pt] (Q1) at (N1.north) {$\le \delta x$};
        \node[small, rounded corners=10pt] (Q2) at (N2.north) {$\le \delta x$};

        \node[med, fill=pink!30, rounded corners=10pt] (TS1) at (S1.north) {$\le x$};
        \node[med, fill=pink!30, rounded corners=10pt] (TS2) at (S2.north) {$\le x$};

        \tikzset{
            brace style/.style={
                decorate, 
                decoration={brace, amplitude=8pt, raise=6pt}, 
                thick
            }
        }

        \draw [brace style] (Q1.north west) -- (Q2.north east) 
            node [midway, above=16pt] {\Large$N_{12}$};

        \draw [brace style] (TS1.north west) -- (TS2.north east) 
            node [midway, above=16pt] {\Large$N_3$};

    \end{tikzpicture}%
    }
    \caption{The structure of the allocation in the proof of \Cref{thm:topn} after the reductions that remove agents with singleton bundles.}
    \label{fig:reduced_alloc2}
\end{figure}

\begin{corollary}\label{cor:23order}
    In ordered instances, Select-and-Eliminate$(N,M)$ returns in polynomial time an allocation which is $2/3$-\gmms.    
\end{corollary}

\begin{remark}\label{rem:23tight}
    A similar construction as in the proof of \Cref{th:upper} can be used to show that the approximation analysis of Select-and-Eliminate is also tight. 
\end{remark}

\subsection{Three agents}

Here we present a variant of Match-Draft-and-Eliminate and an improved analysis tailored to the fundamental case of $n=3$.
The algorithm is identical to \Cref{alg:main}, except it runs the Envy-Cycle-Elimination subroutine in it using a modification of the envy graph defined below. 
This construction is inspired by, but not identical to, the modified envy graph used by \cite{Kurokawa17}.

Let $\mathcal{S}$ be the allocation completed in \hyperref[line:8]{line \ref{line:8}} of \Cref{alg:main}. 
For any partial allocation $\mathcal{B} = (B_1,\ldots,B_n)$ from this point onward, we define the \emph{$\eta$-modified envy graph} $H_{\mathcal{B},\eta} = (N, E_{\mathcal{B},\eta})$, where an edge $(i,j) \in E_{\mathcal{B},\eta}$ if and only if 
\begin{itemize}[leftmargin=20pt,itemsep=4pt]
    \item $v_i(B_j) > (1 + \eta) \cdot v_i(B_i)$, when \textit{all} of the following hold
    \begin{itemize}[leftmargin=25pt,topsep=4pt,itemsep=4pt]
        \item $i \in R$,
        \item $v_i(s_i) < \eta \cdot v_i(f_i)$,
        \item $B_i\supseteq \{f_i, s_i\}$,
        \item $\exists k \in N$ such that $k$'s bundle received a second item in \hyperref[line:4]{lines \ref{line:4}}\hyperref[line:6]{-\ref{line:6}} \emph{before} $i$ received $s_i$ and has not been augmented by Envy-Cycle-Elimination yet;
    \end{itemize}
    \item $v_i(B_j) > v_i(B_i)$, otherwise.
\end{itemize}
In what follows we are going to set the parameter $\eta$ equal to the solution to the equation 
$\frac{2+2\eta}{3+\eta} = \frac{3}{4+\eta}$, i.e.,
$\eta = \frac{\sqrt{57}-7}{4} \approx 0.1375$.
On a high level, for some agents we only add an edge when there is ``significant'' envy. This is to guarantee that an agent who sees someone else having \emph{two} ``somewhat large'' items does not end up with a highly divisible bundle (i.e., a bundle with a large number of ``small'' items from $i$'s perspective) without significantly increasing her value. 
We call \emph{$\eta$-Match-Draft-and-Eliminate} the variant of \Cref{alg:main} where the Envy-Cycle-Elimination  in \hyperref[line:ece]{line \ref{line:ece}} uses the $\eta$-modified envy graph instead. 

\begin{theorem}
\label{thm:three_agents}
Let $\theta = \frac{\sqrt{10}+2}{3}$ and $\eta = \frac{\sqrt{57}-7}{4}$. For any instance with three agents, 
$\eta$-Match-Draft-and-Eliminate$(N, M, \theta)$ outputs a $\frac{\sqrt{10}-1}{3}$-\gmms allocation.
\end{theorem}

We may observe that the last condition for modifying some agent $i$'s outgoing envy edges in $H_{\mathcal{B},\eta}$ is reminiscent of the way agents of the set $N_3$ were defined within the proof of \Cref{lem:other-apx}. Indeed, we need a similar categorization of the agents in $N\setminus\{i\}$. So, given $i\in N$, we partition $N\setminus\{i\}$ to $N_1, N_2, N_3$ as follows: 
\begin{itemize}[leftmargin=20pt,itemsep=4pt]
    \item $N_1$ contains all agents of $N\setminus\{i\}$ whose bundle in $\mathcal{A}$ received an item for the last time during the Envy-Cycle-Elimination step in \hyperref[line:ece]{line \ref{line:ece}};
    \item $N_2$ contains all agents of $N\setminus\{i\}$ whose bundle in $\mathcal{A}$ received its second and last item \emph{after} agent $i$ received $s_i$;
    \item $N_3$ contains all agents of $N\setminus\{i\}$ whose bundle in $\mathcal{A}$ received its second and last item \emph{before} agent $i$ received $s_i$.
\end{itemize}
Like before, by $n_1, n_2, n_3$ we denote the cardinalities of $N_1, N_2, N_3$ respectively. 
Consider a subset $S\subseteq N$ and fix some $i \in S$. 
If necessary, we reduce the subinstance induced by $S$ to get rid of agents with bundles of cardinality $1$ (along with their bundles), as discussed in the beginning of the proof of \Cref{lem:ece-apx}. 
The result is a new instance on $N'\subseteq S $ with $|N'|=n'$, so that there is no $j \in N' \setminus \left\{i \right\}$ such that $|A_j|=1$.
Depending whether $i$ is in $R$ or not, as well as on the values of $n', n_1, n_2, n_3$, we consider a number of distinct cases. 
\Cref{tab:cases} summarizes all the cases along with the corresponding approximation ratio and a pointer to the proposition handling this case.

\begin{table}[h!]
\centering
\begin{tabular}{clccc}
\toprule
\textbf{Case} & \textbf{Description} & \textbf{Factor} & \textbf{Value} ($\theta=\frac{\sqrt{10}+2}{3}$, $\eta=\frac{\sqrt{57}-7}{4}$) & \textbf{Ref.} \\
\midrule
1 & $i \in N\setminus R$, $n' \in [3]$ & $\dfrac{3\theta}{3\theta+2}$ & $\mathbf{0.72076}$ & \Cref{prop:case_L} \\[6pt]
\midrule
2.1 & $n'=2$, $n_1=1$ & ${3}/{4}$ & $0.75$ & \Cref{prop:n2_type1} \\[6pt]
2.2 & $n'=2$, $n_2=1$ & $\dfrac{2}{1+\theta}$ & $0.73509$ & \Cref{prop:n2_type2} \\[6pt]
2.3 & $n'=2$, $n_3=1$ & $\dfrac{2+2\eta}{3+\eta}$ & $0.72508$ & \Cref{prop:n2_type3} \\[6pt]
\midrule
3.1 & $n'=3$, $n_1=2$ & ${3}/{4}$ & $0.75$ & \Cref{prop:n3_case_I} \\[6pt]
3.2 & $n'=3$, $n_2=2$ & $\dfrac{2+\theta}{3\theta}$ & $\mathbf{0.72076}$ & \Cref{prop:n3_case_II} \\[8pt]
3.3 & $n'=3$, $n_3=2$ & $1$ & $1$ & \Cref{prop:n3_case_III} \\[7pt]
3.4 & $n'=3$, $n_1=n_2=1$ & ${3}/{4}$ & $0.75$ & \Cref{prop:n3_case_IV} \\[6pt]
3.5 & $n'=3$, $n_1=n_3=1$ & $\dfrac{3}{4+\eta}$ & $0.72508$ & \Cref{prop:n3_case_V} \\[8pt]
3.6 & $n'=3$, $n_2=n_3=1$ & ${3}/{4}$ & $0.75$ & \Cref{prop:n3_case_VI} \\[2pt]
\midrule
\multicolumn{3}{l}{\textbf{Overall (minimum)}} & $\mathbf{0.72076}$ & \\
\bottomrule
\end{tabular}
\caption{Summary of approximation factors for each case, evaluated at $\theta = \frac{\sqrt{10}+2}{3}\approx 1.7208$ and $\eta = \frac{\sqrt{57}-7}{4} \approx 0.1375$.}
\label{tab:cases}
\end{table}

We denote by $M' = \bigcup_{j \in N'} A_j$ the set of active items in the reduced 
instance. For brevity, we write $\mms_i$ to denote $\mms_i(n', M')$ throughout the 
analysis.

Similarly to the proof of \Cref{th:main}, we begin with the simpler case of $i\notin R$.

\begin{proposition}
\label{prop:case_L}
If $i \in N\setminus R$, then $v_i(A_i) \ge \frac{3\theta}{3\theta+2} \cdot \mms_i$.
\end{proposition}

\begin{proof}
Note that such an agent $i$ is not affected at all by the use of the $\eta$-modified envy graph. Thus, \Cref{lem:ece-apx} applies as is and we get that
$v_i(A_i) \ge \frac{n'\theta}{n'\theta + n' - 1} \cdot \mms_i$.
This expression is decreasing in $n'$, and since $n' \le 3$, the worst case is $n' = 3$, we have 
$v_i(A_i) \ge \frac{3\theta}{3\theta+2} \cdot \mms_i$. 
\end{proof}

When agent $i \in R$, we consider several cases depending on the number and type of the agents in $N'$; see the list below. Recall that since $i \in R$, she has received two items by \hyperref[line:8]{line \ref{line:8}}: $f_i$ with value $v_i(f_i) = x$ and $s_i$ with value $v_i(s_i) = \delta x$ for some $\delta \in [0,1]$. Since the Envy-Cycle-Elimination algorithm only improves an agent's value (even when the $\eta$-modified envy graph is used), we have $v_i(A_i) \ge (1+\delta)x$, which implies $x \le  v_i(A_i) / (1+\delta)$.

If $n' = 2$, there is exactly one other agent $j \in N' \setminus \{i\}$. 
We distinguish three cases based on the type of $j$:
\begin{enumerate}[leftmargin=60pt,itemsep=5pt,label=\textbf{Case 2.\arabic*:}]
    \item $n_1 = 1$, i.e., agent $j\in N_1$ (\Cref{prop:n2_type1}).
    \item $n_2 = 1$, i.e., agent $j\in N_2$ (\Cref{prop:n2_type2}).
    \item $n_3 = 1$, i.e., agent $j\in N_3$ (\Cref{prop:n2_type3}).
\end{enumerate}

If $n' = 3$, we have six cases about the types of the two agents $j, k \in N' \setminus \{i\}$:

\begin{enumerate}[leftmargin=60pt,itemsep=5pt,label=\textbf{Case 3.\arabic*:}]
    \item $n_1 = 2$, i.e., both  $j, k \in N_1$ (\Cref{prop:n3_case_I}).
    \item $n_2 = 2$, i.e., both  $j, k \in N_2$ (\Cref{prop:n3_case_II}).
    \item $n_3 = 2$, i.e., both  $j, k \in N_3$ (\Cref{prop:n3_case_III}).
    \item $n_1 = n_2 = 1$, i.e., one of $j,k$ is in $N_1$ and one in $N_2$ (\Cref{prop:n3_case_IV}).
    \item $n_1 = n_3 = 1$, i.e., one of $j,k$ is in $N_1$ and one in $N_3$ (\Cref{prop:n3_case_V}).
    \item $n_2 = n_3 = 1$, i.e., one of $j,k$ is in $N_2$ and one in $N_3$ (\Cref{prop:n3_case_VI}).
\end{enumerate}

In the propositions that follow, in order to facilitate the presentation, we try to avoid repeating the same arguments we made in the proof of \Cref{lem:other-apx}. For instance, if $i\in R$ and $j\in N_2$, we directly write that $v_i(A_j) \le (\theta + \delta)x$ rather than deriving the inequality from scratch.

\begin{proposition}
\label{prop:n2_type1}
If $i\in R$,  $n' = 2$ and $n_1 = 1$, then
$v_i(A_i) \ge \frac{3}{4} \cdot \mms_i$.
\end{proposition}

\begin{proof}
Since $j\in N_1$, under the use of the $\eta$-modified envy graph it holds that 
\[
v_i(A_j) \le (1+\eta)v_i(A_i) + \delta x.
\]
Indeed, this bound covers both cases: if $j$ received her last item while $i$ had standard outgoing envy edges, then $v_i(A_j) \le v_i(A_i) + \delta x$, which is a stronger bound (since $\eta > 0$). If $j$ received her last item, say $g$, while $i$ had modified outgoing envy edges, then agent $i$ did not $(1+\eta)$-envy $A_j \setminus \{g\}$, i.e., $v_i(A_j \setminus \{g\}) \le (1+\eta)v_i(A_i)$, and since $g$ was in the pool when $i$ chose $s_i$, we have $v_i(g) \le \delta x$. In both cases the bound holds. Therefore,
\begin{align*}
\mms_i &\le \frac{1}{2} v_i(M') \le \frac{1}{2}\left[v_i(A_i) + (1+\eta)v_i(A_i) + \delta x\right] \\
&\le \frac{1}{2}\Big[(2+\eta)v_i(A_i) + \frac{\delta}{1+\delta}v_i(A_i)\Big] = \frac{2+\eta+(3+\eta)\delta}{2+2\delta} \, v_i(A_i).
\end{align*}
By rearranging,
\[
v_i(A_i) \ge \frac{2+2\delta}{2+\eta+(3+\eta)\delta} \cdot \mms_i \ge \frac{4}{5+2\eta} \cdot \mms_i\ge \frac{3}{4} \cdot \mms_i\,,
\]
where the second inequality follows from the fact that, for any $\eta\ge 0$, $f(\delta) = \frac{2+2\delta}{2+\eta+(3+\eta)\delta}$ is decreasing in $\delta$ 
and, thus, is minimized for $\delta = 1$, whereas the third inequality is directly implied by the fact that $0\le \eta \le 1/6$.
\end{proof}

\begin{proposition}
\label{prop:n2_type2}
If $i\in R$, $n' = 2$ and $n_2 = 1$, then $v_i(A_i) \ge \frac{2}{1+\theta} \cdot \mms_i$.
\end{proposition}

\begin{proof}

Since $j\in N_2$, it holds that $v_i(A_j) \le (\theta + \delta)x$. 
We have that

\begin{align*}
\mms_i &\le \frac{1}{2}v_i(M') \le \frac{1}{2}(v_i(A_i) + (\theta+\delta)x) \\
&\le \frac{1}{2}\Big(v_i(A_i) + \frac{\theta+\delta}{1+\delta}\,v_i(A_i)\Big) \le  \frac{1+\theta+2\delta}{2+2\delta} \, v_i(A_i).
\end{align*}
By rearranging,
\[
v_i(A_i) \ge \frac{2+2\delta}{1+\theta+2\delta} \cdot \mms_i \ge \frac{2}{1+\theta} \cdot \mms_i\,,
\]
where the last inequality holds since $\theta \ge 1$ and $\delta \ge 0$. 
\end{proof}

\begin{proposition}
\label{prop:n2_type3}
If $i\in R$, $n' = 2$ and $n_3 = 1$, then $v_i(A_i) \ge \frac{2+2\eta}{3+\eta} \cdot \mms_i$.
\end{proposition}

\begin{proof}
Recall that, since $i\in R$, $x \le  v_i(A_i) / (1+\delta)$; if $\delta \ge \eta$, this implies $x \le  v_i(A_i) / (1+\eta)$. Also, since $j\in N_3$, we have $v_i(A_j) \le 2x$.

We consider two cases, depending on whether $A_i$ contains $\{f_i, s_i\}$ or not. 

\begin{itemize}[leftmargin=20pt,itemsep=5pt]
    \item $A_i \supseteq \{f_i, s_i\}$: 
    First note that, from agent $i$'s 
perspective, the allocation $(\{f_i, s_i\}, \{f_j, s_j\})$ of these four items could equivalently have arisen 
from the picking sequence $ijji$. (That is, $i$ chooses her favorite item, then $j$ chooses her two favorite unallocated items and, finally $i$ gets the last item.) It is known, however, 
that this picking sequence guarantees that the allocation $(\{f_i, s_i\}, \{f_j, s_j\})$ is \mmsT for agent $i$; see, e.g., Lemma 5.3 of \cite{amanatidis2020multiple}. By repeatedly applying 
    Lemma~\ref{lem:aMMSplusg}---if needed---on top of this allocation, we get  $v_i(A_i) \ge \mms_i$. 

\item $A_i \not\supseteq \{f_i, s_i\}$: This means that the bundle containing $\{f_i, s_i\}$ was taken from agent $i$ during a cycle elimination. At the time, if $\delta < \eta$, agent $i$ met all the requirements of only having outgoing edges to agents she $(1+\eta)$-envies. If that was the case, she improved by a factor of $1+\eta$ and we have 
\[
v_i(A_i) \ge (1+\eta) (1+\delta) x \ge (1+\eta) x \implies x \le \frac{1}{1+\eta} v_i(A_i).
\]
Even if this was not the case, i.e., if $\delta \ge \eta$, as discussed in the beginning of the proof, we still have $x \le  v_i(A_i) / (1+\eta)$.

Now, we may bound the maximin share of $i$.
    \begin{align*}
    \mms_i & \le \frac{1}{2}v_i(M')  \le \frac{1}{2}\left[v_i(A_i) + 2x\right] \\ &\le \frac{1}{2}\Big[v_i(A_i) + \frac{2}{1+\eta}\, v_i(A_i)\Big] = \frac{3+\eta}{2(1+\eta)}\, v_i(A_i)
    \end{align*}
    By rearranging, we directly get  $v_i(A_i) \ge \frac{2+2\eta}{3+\eta} \cdot \mms_i$.\qedhere
\end{itemize}
\end{proof}

\begin{proposition}

\label{prop:n3_case_I}
If $i\in R$, $n' = 3$ and $n_1 = 2$, then $v_i(A_i) \ge \frac{3}{4} \cdot \mms_i$.
\end{proposition}

\begin{proof}
We distinguish two cases based on $\delta$:

\begin{itemize}[leftmargin=20pt,itemsep=5pt]
    \item $\delta \ge \eta$: In this case, agent $i$ is not affected  by the use of the $\eta$-modified envy graph. Since both $j$ and $k$ are in $N_1$, we have 
    $v_i(A_j) \le v_i(A_i) + \delta x$ and $v_i(A_k) \le v_i(A_i) + \delta x$. So,

    \[
    \mms_i \le \frac{1}{3}\,v_i(M') \le \frac{1}{3}\left[3v_i(A_i) + 2\delta x\right] \le \frac{3+5\delta}{3+3\delta} \, v_i(A_i)
    \]
    By rearranging,  
    \[
    v_i(A_i) \ge \frac{3+3\delta}{3+5\delta} \cdot \mms_i\ge \frac{3}{4}\cdot \mms_i\,,
    \]
where the last inequality follows from the fact that $f(\delta) = \frac{3+3\delta}{3+5\delta}$ is decreasing in $\delta$ 
and, thus, is minimized for $\delta = 1$.
   \item $\delta < \eta$: Now agent $i$ \emph{may be} affected  by the use of the $\eta$-modified envy graph. 

    Therefore, for both $j, k\in N_1$ we have 
    $v_i(A_j) \le (1+\eta)v_i(A_i) + \delta x$ and $v_i(A_k) \le (1+\eta)v_i(A_i) + \delta x$; see also the proof of \Cref{prop:n2_type1}. Thus,

    \[
    \mms_i \le \frac{1}{3}\,v_i(M')\le \frac{1}{3}\left[(3+2\eta)v_i(A_i) + 2\delta x\right] \le \frac{3+2\eta+(5+2\eta)\delta}{3+3\delta} \, v_i(A_i).
    \]
    By rearranging,  
    \[
    v_i(A_i) \ge \frac{3+3\delta}{3+2\eta+(5+2\eta)\delta} \,\mms_i \ge \frac{3+3\eta}{3+2\eta+(5+2\eta)\eta} \,\mms_i > \frac{3}{4} \,\mms_i \,,
    \] 
where the second inequality follows from the fact that, for any $\eta\ge 0$, $f_1(\delta) = \frac{3+3\delta}{3+2\eta+(5+2\eta)\delta}$ is decreasing in $\delta$ 
and, thus, is minimized for $\delta = \eta$, whereas the third inequality follows from the fact that $f_2(\eta) = \frac{3+3\eta}{3+2\eta+(5+2\eta)\eta}$ is decreasing in $\eta$ 
and, given that $\eta \le 1/6$, a bound is $f_2(0.17)\approx 0.82 >3/4$.\qedhere
\end{itemize}
\end{proof}

\begin{proposition}
\label{prop:n3_case_II}
If $i\in R$, $n' = 3$ and $n_2 = 2$, then $v_i(A_i) \ge \frac{2+\theta}{3\theta} \cdot \mms_i$.
\end{proposition}

\begin{proof}
Like \Cref{prop:n2_type3}, this is one of the cases where we need to resort to the definition of the maximin share, rather than bounding it by the average value. 
Since $n_2 = 2$, both $j, k \in N_2$: each has received two items such that  
$\max\{v_i(f_j), v_i(f_k)\} \le \theta x$ and $\max\{v_i(s_j), v_i(s_k)\} \le \delta x$.
From agent $i$'s perspective, this allocation $(\{f_i, s_i\}, \{f_j, s_j\}, \{f_k, s_k\})$ of these six items could equivalently have been the result of choosing according to the picking sequence $jkiikj$. (That is, $j$ chooses her favorite item, then $k$ chooses her favorite unallocated item, and so on.)

Let $T = \{g_1, \dots, g_6\}$ be these goods, indexed according to agent $i$'s preference with $g_1$ being her favorite among them.
\begin{itemize}[leftmargin=20pt,itemsep=5pt]
    \item If $\mms_i(3, T) > \theta x$: In any $3$-maximin share defining partition for $i$, $g_1$ and $g_2$ 
    are never singletons.
    Since agent $i$ sees 
    the item values as $y\le\theta x, z\le \theta x, x, \delta x, w\le \delta x, t \le \delta x$, it holds that in any partition 
    the worst bundle has value at most $(1+\delta)x$. 
    Hence, $\mms_i(3,T) \le (1+\delta)x \le v_i(\{f_i, s_i\})$, i.e., the approximation factor in this case is $1$.
    \item If $\mms_i(3, T) \le \theta x$: We consider cases based on the potential total value of items $g_3, \dots, g_6$, which is $1+3\delta$.
    \begin{itemize}[leftmargin=20pt,topsep=5pt,itemsep=5pt]
        \item $1+3\delta < \theta$: Then $\mms_i(3,T) \le (1+3\delta)x$ but also $\delta < {(\theta-1)}/{3}$. Therefore,
        \begin{align*}
        v_i(\{f_i, s_i\}) &= (1+\delta)\,x = \frac{1+\delta}{1+3\delta} \, (1+3\delta) \, x \ge \frac{1+\delta}{1+3\delta} \,\mms_i(3,T) \\
        &\ge \frac{1+\frac{\theta-1}{3}}{1+3\cdot\frac{\theta-1}{3}}\cdot \mms_i(3,T) = \frac{2+\theta}{3\theta}\cdot \mms_i(3,T)\,,
        \end{align*}
        where the last inequality holds since $f(\delta) = \frac{1+\delta}{1+3\delta}$ is decreasing in $\delta$, so we can substitute the upper bound for $\delta$, namely $\frac{\theta-1}{3}$, to obtain a valid lower bound.
        \item $1+3\delta \ge \theta$: Then $\mms_i(3,T) \le \theta x$ and $\delta \ge \frac{\theta-1}{3}$. Therefore,
        \begin{equation*}
        v_i(\{f_i, s_i\}) = (1+\delta)\,x = \frac{1+\delta}{\theta} \, \theta \, x \ge \frac{1+\frac{\theta-1}{3}}{\theta} \cdot \mms_i(3,T) = \frac{2+\theta}{3\theta} \cdot \mms_i(3,T)\,,
        \end{equation*}
        where the last inequality holds since $f(\delta) = \frac{1+\delta}{\theta}$ is increasing in $\delta$, so it is minimized by setting $\delta = \frac{\theta-1}{3}$.
    \end{itemize}
\end{itemize}
In all cases $v_i(\{f_i, s_i\}) \ge \frac{2+\theta}{3\theta} \cdot \mms_i(3,T)$. By repeatedly applying  Lemma~\ref{lem:aMMSplusg} with $\alpha = \frac{2+\theta}{3\theta}$, if necessary, we conclude that $v_i(A_i) \ge \frac{2+\theta}{3\theta} \cdot \mms_i(n',M')$.
\end{proof}

\begin{proposition}
\label{prop:n3_case_III}
If $i\in R$, $n' = 3$ and $n_3 = 2$, then  
$
v_i(A_i) \ge  \mms_i
$.
\end{proposition}

\begin{proof}
Again, we directly argue about the maximin share, instead of resorting to the average value. 
Since $n_3 = 2$, both $j, k\in N_3$, meaning their items satisfy
$
\max\{v_i(f_j), v_i(s_j), v_i(f_k), v_i(s_k)\} \le x
$.

From agent $i$'s perspective, the allocation $(\{f_i, s_i\}, \{f_j, s_j\}, \{f_k, s_k\})$ of these six items could equivalently have been the result of choosing according to the picking sequence $ijkkji$. (That is, $i$ chooses her favorite item, then $j$ chooses her favorite unallocated item, and so on.) Let $T = \{g_1, \dots, g_6\}$ be a renaming of the goods in $\{f_i, s_i, f_j, s_j, f_k, s_k\}$, indexed according to agent $i$'s preference with $g_1$ being her favorite among them.
Consider a $3$-maximin share defining partition $\mathcal{C}$ for agent $i$. Then, either every 
bundle in the partition contains at least two items, or there exists at least 
one a singleton.
\begin{itemize}[leftmargin=20pt,itemsep=5pt]
    \item Suppose there exists a singleton bundle $\{g\}$ in $\mathcal{C}$. 
    Then, $\mms_i(3, T) \le v_i(g)$. Since $g_1$ is the most preferred 
    item of $i$ in $T$, we have $v_i(g) \le v_i(g_1) = v_i(f_i)$. Thus, $\mms_i(3, T) \le v_i(f_i)  \le v_i(\{f_i, s_i\})$.
    
    \item Suppose that every bundle in $\mathcal{C}$ contains exactly two items 
    of $T$. One of the bundles in $\mathcal{C}$ is of the form $\{g_6, g_y\}$, where $y\in[5]$. Since $g_1$ is the most valuable item of $i$ in $T$, pairing 
    $g_6$ with $g_1$ yields the highest possible value for such a bundle. Also, recall that $v_i(g_1) = v_i(f_i)$ and $v_i(g_6) \le v_i(s_i)$. Therefore,
    $\mms_i(3, T) \le v_i(\{g_6, g_y\}) \le v_i(\{g_6, g_1\}) \le v_i(\{f_i, s_i\})$.
\end{itemize}
In any case, $v_i(\{f_i, s_i\}) \ge \mms_i(3, T)$. By repeatedly applying  Lemma~\ref{lem:aMMSplusg}, if necessary, we conclude that $v_i(A_i) \ge  \mms_i(n',M')$. \end{proof}

\begin{proposition}
\label{prop:n3_case_IV}
If $i\in R$, $n' = 3$ and $n_1 = n_2 = 1$, then $v_i(A_i) \ge \frac{3}{4} \cdot \mms_i$.
\end{proposition}

\begin{proof}
Without loss of generality, assume that $j\in N_1$ and $k\in N_2$. The upper bound $v_i(A_k)\le (\theta +\delta) \, x$ we have  derived before (see, e.g., \Cref{prop:n2_type2}) holds in any case, but upper bounding  $v_i(A_j)$ depends on the value of $\delta$ (see also \Cref{prop:n3_case_I}):

\begin{itemize}[leftmargin=20pt,itemsep=5pt]
    \item If $\delta \ge \eta$, then agent $i$ is not affected  by the use of the $\eta$-modified envy graph and $v_i(A_j) \le v_i(A_i) + \delta x$. Thus, we get
    \[
    \mms_i \le \frac{1}{3}\,v_i(M') \le \frac{1}{3}\left[2v_i(A_i) + \delta x + (\theta+\delta)x\right] \le \frac{2+\theta+4\delta}{3+3\delta} v_i(A_i). 
    \]

By rearranging, 
\[v_i(A_i) \ge \frac{3+3\delta}{2+\theta+4\delta}\cdot \mms_i \ge \frac{6}{6+\theta}\cdot \mms_i \ge \frac{3}{4}\cdot \mms_i\,,\]
where the second inequality holds  because, for any $\theta\in[1,1.8]$, $f(\delta) = \frac{3+3\delta}{2+\theta+4\delta}$ is decreasing in $\delta$ 
and, thus, is minimized for $\delta = 1$, while the third just uses the fact that $\theta \le 2$.

    \item If $\delta < \eta$, then agent $i$ may be affected  by the use of the $\eta$-modified envy graph and  $v_i(A_j) \le (1+\eta)v_i(A_i) + \delta x$. Thus, we have
    \[
    \mms_i \le \frac{1}{3}\left[(2+\eta)v_i(A_i) + \delta x + (\theta+\delta)x\right] \le \frac{2+\theta+\eta+(4+\eta)\delta}{3+3\delta}\, v_i(A_i).
    \]
    By rearranging, 
    \[v_i(A_i) \ge \frac{3+3\delta}{2+\theta+\eta+(4+\eta)\delta} \cdot \mms_i \ge \frac{126}{103+ 36\,\theta} \cdot \mms_i \ge \frac{3}{4} \cdot \mms_i,
    \]

where the second inequality follows from the fact that, for any $\theta\in [1,1.8], \eta\ge 0$, $f_(\delta) = \frac{3+3\delta}{2+\theta+\eta+(4+\eta)\delta}$ is decreasing in $\delta$ and $\eta$ 
and, thus, is minimized for $\delta = \eta = 1/6$, whereas the third inequality just uses the fact that $\theta\le 1.8$.\qedhere
\end{itemize}
\end{proof}

\begin{proposition}
\label{prop:n3_case_V}
If $i\in R$, $n' = 3$ and $n_1 = n_3 = 1$, then $v_i(A_i) \ge \frac{3}{4+\eta} \cdot \mms_i$.
\end{proposition}

\begin{proof}
Without loss of generality, assume that $j\in N_1$ and $k\in N_3$. The upper bound $v_i(A_k)\le 2 x$ we have  derived before (see, e.g., \Cref{prop:n2_type3}) holds in any case, but upper bounding  $v_i(A_j)$ depends on the value of $\delta$ (see also \Cref{prop:n3_case_I}):

\begin{itemize}[leftmargin=20pt,itemsep=5pt]
    \item If $\delta \ge \eta$, then agent $i$ is not affected  by the use of the $\eta$-modified envy graph and $v_i(A_j) \le v_i(A_i) + \delta x$. Thus, we get
    \[\mms_i \le \frac{1}{3}\,v_i(M') \le \frac{1}{3}\left[2v_i(A_i) + \delta x + 2x\right] \le \frac{4+3\delta}{3+3\delta} v_i(A_i). 
    \]
    By rearranging, 
\[v_i(A_i) \ge \frac{3+3\delta}{4+3\delta}\cdot \mms_i \ge \frac{3}{4}\cdot \mms_i\,,\]
where the last inequality trivially holds for any $\delta \ge 0$.
\item If $\delta < \eta$, then agent $i$ may be affected  by the use of the $\eta$-modified envy graph and  $v_i(A_j) \le (1+\eta)v_i(A_i) + \delta x$. Thus, we have
    \[
    \mms_i \le \frac{1}{3}\left[(2+\eta)v_i(A_i) + \delta x + 2x\right] \le \frac{4+\eta+(3+\eta)\delta}{3+3\delta}\, v_i(A_i).
    \]
    By rearranging, 
    \[v_i(A_i) \ge \frac{3+3\delta}{4+\eta+(3+\eta)\delta} \cdot \mms_i \ge \frac{3}{4+\eta} \cdot \mms_i,
    \]
    where the last inequality follows by observing that, for any $\eta\ge 0$, $f(\delta) = \frac{3+3\delta}{4+\eta+(3+\eta)\delta}$ is increasing in $\delta$ 
and, thus, is minimized for $\delta = 0$.
\end{itemize}
Since $\frac{3}{4+\eta} \le \frac{3}{4}$, the overall approximation ratio in this case is $\frac{3}{4+\eta}$. 
\end{proof}

\begin{proposition}
\label{prop:n3_case_VI}
If $i\in R$, $n' = 3$ and $n_2 = n_3 = 1$, then $v_i(A_i) \ge \frac{3}{4} \cdot \mms_i$.
\end{proposition}

\begin{proof}
Once again, we are going to argue about the maximin share directly.  
Without loss of generality, assume that $j\in N_2$ and $k\in N_3$.
This means that $v_i(f_j)\le \theta x$, $v_i(s_j)\le \delta x$ and $\max\{v_i(f_k), v_i(s_k)\} \le x$.

From agent $i$'s perspective, the allocation $(\{f_i, s_i\}, \{f_j, s_j\}, \{f_k, s_k\})$ of these six items could have been the result of choosing according to the picking sequence $jikkij$. (That is, $j$ chooses her favorite item, then $i$ chooses her favorite unallocated item, and so on.) As usual, let $T = \{g_1, \dots, g_6\}$ be a renaming of the goods in $\{f_i, s_i, f_j, s_j, f_k, s_k\}$, indexed according to agent $i$'s preference with $g_1$ being her favorite among them.

\begin{itemize}[leftmargin=20pt,itemsep=5pt]
    \item If $\mms_i(3, T) > \theta x$: In any $3$-maximin share defining partition, $g_1$ is never 
    a singleton. With only four available items among $g_2, \ldots, g_6$ remaining (which agent $i$  values as $ y\le x, z\le x, w \le x, \delta x, t \le \delta x$), we cannot have 
    two bundles with value more than $(1+\delta)x$. Therefore, $\mms_i(3,T) \le (1+\delta)x \le v_i(\{f_i, s_i\}) $, i.e., the approximation ratio here is $1$.

    \item If $\mms_i(3, T) \le \theta x$: 
    If there exists a $3$-maximin share defining partition 
    where $g_1$ is not a singleton, then we can repeat the analysis of the previous case to obtain an approximation ratio of $1$.
    So, we may assume that in any $3$-maximin share defining partition $\{g_1\}$ is one of the bundles. Then, by  \Cref{lem:monotonicity}, we have
    \[
    \mms_i(3, T) \le \mms_i(2, T \setminus \{g_1\}).
    \]
    However, it is easy to see that no matter how we partition the items $g_2, \ldots, g_6$ into two sets, it is not possible that both have value strictly more than $\min\{2x,\, (1+2\delta)x\}$. That is, $\mms_i(2, T \setminus \{g_1\}) \le \min\{2x,\, (1+2\delta)x\}$. We consider two cases depending on the value of $\delta$.
    \begin{itemize}[leftmargin=20pt,topsep=5pt,itemsep=5pt]
        
        \item $\delta < \frac{1}{2}$: Then $\mms_i(3,T) \le (1+2\delta)x$. Therefore,
        \[
        v_i(\{f_i, s_i\})  = (1+\delta)x = \frac{1+\delta}{1+2\delta} \, (1+2\delta)\, x 
        \ge \frac{1+\delta}{1+2\delta} \cdot \mms_i(3, T) 
        \ge \frac{3}{4} \cdot \mms_i(3, T),
        \]
        where the last inequality holds since $f(\delta) = \frac{1+\delta}{1+2\delta}$ is decreasing in $\delta$, so we can substitute the upper bound for $\delta$, namely $1/2$, to obtain a valid lower bound.
        \item $\delta \ge \frac{1}{2}$: Then $\mms_i(3,T) \le 2x$. Therefore,
        \[
        v_i(\{f_i, s_i\})  = (1+\delta)x = \frac{1+\delta}{2} \, 2x 
        \ge \frac{1+\delta}{2} \cdot \mms_i(3, T) 
        \ge \frac{3}{4} \cdot \mms_i(3, T),
        \]
        where the last inequality holds since $f(\delta) = \frac{1+\delta}{2}$ is increasing in $\delta$, so it is minimized by setting $\delta = 1/2$.
    \end{itemize}
\end{itemize}
In all cases $v_i(\{f_i, s_i\}) \ge \frac{3}{4} \cdot \mms_i(3,T)$. By repeatedly applying  \Cref{lem:aMMSplusg} with $\alpha = 3/4$, if necessary, we conclude that $v_i(A_i) \ge \frac{3}{4} \cdot \mms_i(n',M')$.
\end{proof}

Finally, we can complete the proof of  \Cref{thm:three_agents}.

\begin{proof}[Proof of \Cref{thm:three_agents}]
The theorem follows directly by \hyperref[prop:case_L]{Propositions \ref{prop:case_L}}\hyperref[prop:n3_case_VI]{-\ref{prop:n3_case_VI}} which directly correspond to the case analysis summarized in \Cref{tab:cases}, by setting $\theta = \frac{\sqrt{10}+2}{3} \in [1,1.8]$ and $\eta = \frac{\sqrt{57}-7}{4}\in [0,1/6]$. We notice that the minimum approximation factor across all cases is $\frac{\sqrt{10}-1}{3}$, achieved in \Cref{prop:case_L,prop:n3_case_II}.
\end{proof}

\section{Discussion}
In this work we show that $(\phi-1)$-\gmms allocations always exist and can be efficiently computed, for any number of agents with additive valuation functions. For small values of $n$ or for instances where the agents agree about their top-$n$ goods, we obtain significant improvements. This is the first improvement on \gmms allocations since the works of \cite{chaudhury2021little} and \cite{amanatidis2020multiple}, indicating how challenging the problem is. A main reason for this is that the notion of \gmms fairness seems to combine characteristics of both share-based and envy-based notions, without fitting well with either of them. Besides the obvious direction of improving our results (computationally or existentially), showing upper bounds, or extending the study of \gmms fairness beyond the additive domain, there is a less well-defined but very intriguing direction. There is a strong underlying connection between the notions of \gmms and \efx: approximately \efx allocations are approximately \gmms \citep{AmanatidisBM18}, exact \gmms allocations are \efx (albeit for a slightly weaker version of \efx) \citep{BBMN18}, the state-of-the-art approximation algorithms for the two notions are the same and, as it turns out now, with the same exact approximation ratio. Given the significance of \efx-related questions in the recent development of the fair division literature, we believe that the connection between the two notions should be further explored.  

\section*{Acknowledgements}
This work was supported by the project MIS 5154714 of the National Recovery and Resilience Plan “Greece 2.0” funded by the European Union under the NextGenerationEU Program.

\bibliography{phi-gmms-references.bib}

@article{amanatidis2020multiple,
  title={Multiple birds with one stone: Beating 1/2 for {EFX} and {GMMS} via envy cycle elimination},
  author={Amanatidis, Georgios and Markakis, Evangelos and Ntokos, Apostolos},
  journal={Theoretical Computer Science},
  volume={841},
  pages={94--109},
  year={2020},
  publisher={Elsevier}
}

@article{HuangZ25,
  author       = {Xin Huang and
                  Shengwei Zhou},
  title        = {An {FPTAS} for 7/9-Approximation to Maximin Share Allocations},
  journal      = {CoRR},
  volume       = {abs/2511.13056},
  year         = {2025},
  url          = {https://doi.org/10.48550/arXiv.2511.13056},
  xdoi          = {10.48550/ARXIV.2511.13056},
  eprinttype   = {arXiv},
}

@inproceedings{AmanatidisBM18,
	author    = {Georgios Amanatidis and
	Georgios Birmpas and
	Vangelis Markakis},
	xeditor    = {J{\'{e}}r{\^{o}}me Lang},
	title     = {Comparing Approximate Relaxations of Envy-Freeness},
	booktitle = {Proceedings of the Twenty-Seventh International Joint Conference on
	Artificial Intelligence, {IJCAI} 2018},
	pages     = {42--48},
	publisher = {ijcai.org},
	year      = {2018},
	xurl       = {https://doi.org/10.24963/ijcai.2018/6},
}

@InProceedings{MMP25,
author="Mancho, Alviona
and Markakis, Evangelos
and Protopapas, Nicos",
xeditor="Lavi, Ron
and Zhang, Jie",
title="Fairness Under Equal-Sized Bundles: Impossibility Results and Approximation Guarantees",
booktitle="18th International Symposium on Algorithmic Game Theory, {SAGT} 2025",
year="2025",
pages="191--208",
}

@inproceedings{MS23,
author = {Markakis, Evangelos and Santorinaios, Christodoulos},
title = {Improved {EFX} Approximation Guarantees under Ordinal-based Assumptions},
year = {2023},
booktitle = {Proceedings of the 22nd International Conference on Autonomous Agents and Multiagent Systems, {AAMAS} 2023},
pages = {591–599},
}

@article{AR26,
      title={Simultaneous Ordinal Maximin Share and Envy-Based Guarantees}, 
      author={Hannaneh Akrami and Timo Reichert},
      year={2026},
      journal      = {CoRR},
  volume       = {abs/2511.13056},
      url={https://arxiv.org/abs/2602.15566},
        eprinttype   = {arXiv},
}

@inproceedings{ChristoforidisS24,
  author       = {Vasilis Christoforidis and
                  Christodoulos Santorinaios},
  title        = {On the Pursuit of {EFX} for Chores: Non-existence and Approximations},
  booktitle    = {Proceedings of the Thirty-Thirds International Joint Conference on
                  Artificial Intelligence, {IJCAI} 2024},
  pages        = {2713--2721},
  publisher    = {ijcai.org},
  year         = {2024},
  xurl          = {https://www.ijcai.org/proceedings/2024/300},
}

@inproceedings{Farhadietal21,
  author       = {Alireza Farhadi and
                  Mohammad Taghi Hajiaghayi and
                  Mohamad Latifian and
                  Masoud Seddighin and
                  Hadi Yami},
  title        = {Almost Envy-freeness, Envy-rank, and {N}ash Social Welfare Matchings},
  booktitle    = {35th {AAAI} Conference on Artificial Intelligence, {AAAI}
                  2021},
  pages        = {5355--5362},
  publisher    = {{AAAI} Press},
  year         = {2021},
  xurl          = {https://doi.org/10.1609/aaai.v35i6.16675},
}

@phdthesis{Kurokawa17,
	title    = {Fair Division in Game Theoretic Settings},
	school   = {Carnegie Mellon University},
	author   = {David Kurokawa},
	year     = {2017}
}

@inproceedings{LMMS04,
	author    = {Richard J. Lipton and
	Evangelos Markakis and
	Elchanan Mossel and
	Amin Saberi},
	IGNOREeditor    = {Jack S. Breese and
	Joan Feigenbaum and
	Margo I. Seltzer},
	title     = {On approximately fair allocations of indivisible goods},
	booktitle = {Proceedings 5th {ACM} Conference on Electronic Commerce (EC-2004)},
	pages     = {125--131},
	publisher = {{ACM}},
	year      = {2004},
	xurl       = {https://doi.org/10.1145/988772.988792},
}

@article{chaudhury2021little,
  title={A little charity guarantees almost envy-freeness},
  author={Chaudhury, Bhaskar Ray and Kavitha, Telikepalli and Mehlhorn, Kurt and Sgouritsa, Alkmini},
  journal={SIAM Journal on Computing},
  volume={50},
  number={4},
  pages={1336--1358},
  year={2021},
  publisher={SIAM}
}

@article{CaragiannisKMPS19,
	author    = {Ioannis Caragiannis and
	David Kurokawa and
	Herv{\'{e}} Moulin and
	Ariel D. Procaccia and
	Nisarg Shah and
	Junxing Wang},
	title     = {The Unreasonable Fairness of Maximum {N}ash Welfare},
	journal   = {{ACM} Trans. Economics and Comput.},
	volume    = {7},
	number    = {3},
	pages     = {12:1--12:32},
	year      = {2019},
}

@article{Budish11,
	Author = {Eric Budish},
	Journal = {Journal of Political Economy},
	Number = {6},
	Pages = {1061-1103},
	Title = {The Combinatorial Assignment Problem: Approximate Competitive Equilibrium from Equal Incomes},
	Volume = {119},
	Year = {2011}}

@inproceedings{BBMN18,
	author    = {Siddharth Barman and
	Arpita Biswas and
	Sanath Kumar Krishna Murthy and
	Yadati Narahari},
	IGNOREeditor    = {Sheila A. McIlraith and
	Kilian Q. Weinberger},
	title     = {Groupwise Maximin Fair Allocation of Indivisible Goods},
	booktitle = {Proceedings of the Thirty-Second {AAAI} Conference on Artificial Intelligence,
	(AAAI-18)},
	pages     = {917--924},
	publisher = {{AAAI} Press},
	year      = {2018},
	xurl       = {https://www.aaai.org/ocs/index.php/AAAI/AAAI18/paper/view/16856},
}

@article{BL16,
	author    = {Sylvain Bouveret and
	Michel Lema{\^{\i}}tre},
	title     = {Characterizing conflicts in fair division of indivisible goods using
	a scale of criteria},
	journal   = {Autonomous Agents and Multi-Agent Systems},
	volume    = {30},
	number    = {2},
	pages     = {259--290},
	year      = {2016},
	xurl       = {https://doi.org/10.1007/s10458-015-9287-3},
}

@book{GS58,
	Author = {George Gamow and Marvin Stern},
	Publisher = {Viking press},
	Title = {Puzzle-Math},
	Year = {1958}}

@article{Foley67,
	Author = {Duncan K. Foley},
	Journal = {Yale Economics Essays},
	Pages = {45-98},
	Title = {Resource Allocation and the Public Sector},
	Volume = {7},
	Year = {1967}}

@article{Varian74,
	Author = {Hal R. Varian},
	Journal = {Journal of Economic Theory},
	Pages = {63-91},
	Title = {Equity, Envy and Efficiency},
	Volume = {9},
	Year = {1974}}

@article{kurokawa2018fair,
  title={Fair enough: Guaranteeing approximate maximin shares},
  author={Kurokawa, David and Procaccia, Ariel D and Wang, Junxing},
  journal={Journal of the ACM (JACM)},
  volume={65},
  number={2},
  pages={1--27},
  year={2018},
  publisher={ACM New York, NY, USA}
}

@inproceedings{HeidariKSS26,
  author       = {Ehsan Heidari and
                  Alireza Kaviani and
                  Masoud Seddighin and
                  AmirMohammad Shahrezaei},
  xeditor       = {Kasper Green Larsen and
                  Barna Saha},
  title        = {Improved Maximin Share Guarantee for Additive Valuations},
  booktitle    = {Proceedings of the 2026 Annual {ACM-SIAM} Symposium on Discrete Algorithms,
                  {SODA} 2026},
  pages        = {2239--2290},
  publisher    = {{SIAM}},
  year         = {2026},
  xurl          = {https://doi.org/10.1137/1.9781611978971.81},
}

@inproceedings{feige2021tight,
  title={A tight negative example for {MMS} fair allocations},
  author={Feige, Uriel and Sapir, Ariel and Tauber, Laliv},
  booktitle={International Conference on Web and Internet Economics},
  pages={355--372},
  year={2021},
  organization={Springer}
}

@inproceedings{PR18,
	author    = {Benjamin Plaut and
	Tim Roughgarden},
	IGNOREeditor    = {Artur Czumaj},
	title     = {Almost Envy-Freeness with General Valuations},
	booktitle = {Proceedings of the Twenty-Ninth Annual {ACM-SIAM} Symposium on Discrete
	Algorithms, {SODA} 2018},
	pages     = {2584--2603},
	publisher = {{SIAM}},
	year      = {2018},
	xurl       = {https://doi.org/10.1137/1.9781611975031.165},
}

@article{AmanatidisABFLMVW23,
  author       = {Georgios Amanatidis and
                  Haris Aziz and
                  Georgios Birmpas and
                  Aris Filos{-}Ratsikas and
                  Bo Li and
                  Herv{\'{e}} Moulin and
                  Alexandros A. Voudouris and
                  Xiaowei Wu},
  title        = {Fair division of indivisible goods: Recent progress and open questions},
  journal      = {Artif. Intell.},
  volume       = {322},
  pages        = {103965},
  year         = {2023},
  xurl          = {https://doi.org/10.1016/j.artint.2023.103965},
}

@article{Steinhaus49,
	Author = {Hugo Steinhaus},
	Journal = {Econometrica},
	Pages = {315-319},
	Title = {Sur la division pragmatique},
	Volume = {17 (Supplement)},
	Year = {1949}
}

\newpage
\crefalias{section}{appendix}
\appendix
\renewcommand{\theHsection}{appendix.\Alph{section}}

\section{The Envy Cycle Elimination Algorithm}\label{app:ece}

Here we state a version of the celebrated the \emph{envy cycle elimination algorithm} of \citet{LMMS04} (\Cref{alg:ece}) that runs on top of any partial allocation, as  needed for \Cref{alg:main}.
This algorithm extends an allocation one good at a time. In each step, an agent that no one envies receives the next available good. To ensure that such an agent always exists, the algorithm identifies  cycles in the envy graph and eliminates them by reallocating some of the current bundles along these cycles.
In \Cref{thm:xece} below, which subsumes \Cref{thm:ece} from \Cref{sec:prelims}, we also summarize the main known properties of \Cref{alg:ece}.

\begin{algorithm}[ht]
	\DontPrintSemicolon 
   {\normalsize 
		\For{every $g \in M'$ in lexicographic order}{
			\While{there is no node of in-degree 0 in $G_{\mathcal{A}}$ }{
				Find a cycle $j_1 \to j_2 \to \ldots \to j_r \to j_1$ in $G_{\mathcal{A}}$ \;
				$B = A_{j_1}$\;
				\For{$k=1$ to $r-1$}{
					$A_{j_k}=A_{j_{k+1}}$ \tcc*{{\small shift the bundles along the cycle}}
				}
				$A_{j_r}=B$\;
			}
			Let $i\in N$ be the lexicographically first node of in-degree $0$ \label{line:source} \; $A_{i} = A_{i}\cup \{g\}$\;
		}
		\Return $\mathcal{A}$ \; 
	}
	\caption{Envy-Cycle-Elimination$(N, \mathcal{A}, M')$\newline {\small $N$: set of agents, $\mathcal{A}$: initial partial allocation, $M'$: set of unallocated goods}} \label{alg:ece}
\end{algorithm}

\begin{theorem}[Follows by \cite{LMMS04}]\label{thm:xece}
	Let $\mathcal{A}$ be any \efo partial allocation and $M'= M \setminus\cup_{i=1}^{n}A_i$. Then, 
	\begin{enumerate}[leftmargin=*,itemsep=3pt,topsep=2pt,parsep=0pt,partopsep=0pt,label=\rm{\alph*})]
		\item at the end of each iteration of the \emph{for} loop, the resulting partial allocation  is \efo. Hence, the algorithm terminates with an \efo allocation in polynomial time. This holds even for agents with general monotone valuation functions.
		\item Fix an agent $i$, and let $A_i$ be the bundle assigned to $i$ at the beginning of some iteration of the outer \emph{for} loop. If $A_i'$ is  assigned to $i$ at the end of any future iteration, then $v_i(A_i') \geq v_i(A_i)$.\label{fact:value}
	\end{enumerate}
\end{theorem}

The first property of \Cref{thm:xece} simply says that the \efo property is maintained during the course of the algorithm, given an initial \efo allocation. The second property states that agents never see the value of their bundle decrease throughout the execution of the algorithm. 
\end{document}